\documentclass[11pt]{article}
\pdfoutput=1
\usepackage{euscript,amsmath,amsthm,amsfonts,amssymb}
\usepackage{mathtools}
\usepackage{subfig}
\usepackage{float}
\usepackage[margin=1in]{geometry}
\usepackage{graphicx}
\usepackage{outlines}
\usepackage[dvipsnames]{xcolor}
\usepackage{jheppub}
\newcommand{\ket}[1]{|{#1}\rangle}

\DeclareMathOperator{\area}{area}

\DeclareMathOperator{\Tr}{Tr}

\newtheorem{theorem}{Theorem}
\newtheorem{definition}{Definition}

\title{\boldmath Bit threads and holographic entanglement of purification}
\author{Jonathan Harper and Matthew Headrick}
\affiliation{Martin Fisher School of Physics, Brandeis University, Waltham, Massachusetts 02453, USA}
\preprint{BRX-TH-6642}
\emailAdd{jharper@brandeis.edu}
\emailAdd{headrick@brandeis.edu}
\abstract{Generalizing the bit thread formalism, we use convex duality to derive dual flow programs to the bipartite and multipartite holographic entanglement of purification proposals and then prove several inequalities using these constructions. In the multipartite case we find the flows exhibit novel behavior which allows for a constrained flux on the boundary of the homology region. We show this flux can be made distinct from bipartite terms and reflects the truly multipartite portion of the holographic entanglement of purification.}
\begin{document}
\maketitle
\flushbottom

\section{Introduction}

The study of geometry has been key to our current understanding of quantum gravity. This is most strongly evidenced by the Ryu-Takayanagi (RT) formula, which allows one to calculate the entanglement entropy as the area of a bulk minimal surface for systems with holographic duals (in static states or states with time-reflection symmetry). The entanglement entropy is the canonical measure of entanglement for pure states. Recent work has attempted to define a similar measure of entanglement for a mixed state on a bipartite region $AB$ as the cross section $E_{w}(A:B)$ of the joint homology region $r(AB)$, the bulk region bounded by $AB$ and its RT surface $m(AB)$ (see fig.\ \ref{fig:bipEOP}).\footnote{For this paper we will only be working with the static or time reflection-symmetric spacetimes, where our slice is a moment of time symmetry. The homology region $r(AB)$ is a slice of the entanglement wedge $E_{W}[AB]$. As we will briefly discuss in sec.\ \ref{sec4}, the conjecture has a natural generalization to the covariant case.} 
More precisely, $E_w(A:B)$ is the area of the minimal surface in $r(AB)$ homologous to $A$ relative to $m(AB)$; we will denote this surface $b_p(A:B)$.

\begin{figure}[hbt]
\centering  
\includegraphics[width=.4\textwidth]{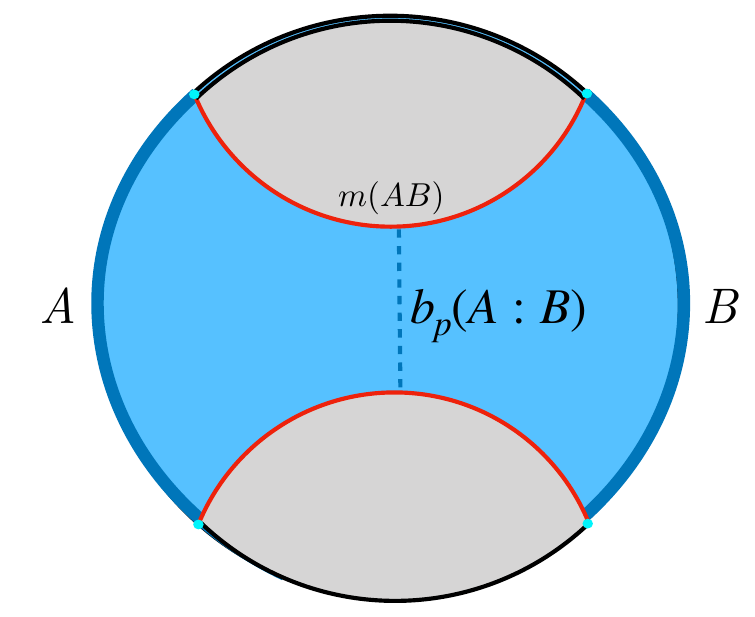}
\caption{\label{fig:bipEOP}
Illustration of the homology-region cross section $E_w(A:B)$, conjectured to be dual to the entanglement of purification \eqref{EOPdef}. The $AB$ homology region $r(AB)$ is shown in blue. $E_w(A:B)$ is defined as the area of the minimal surface $b_p(A:B)$ in $r(AB)$ homologous to $A$ relative to $m(AB)$.
}
\end{figure}

This geometric quantity is conjectured to be dual to a certain measure of quantum and classical correlations in the boundary theory, namely the entanglement of purification (EOP) \cite{Umemoto_2018,Nguyen2018}.\footnote{Several competing conjectures have also appeared in the literature. In \cite{Dutta:aa} the homology-region cross section was calculated to be equal to the reflected entropy: half the entanglement entropy of a canonical purification whose dual geometry is a wormhole formed by gluing two copies of the homology region. In \cite{Kudler-Flam:2018aa} the logarithmic negativity was proposed as a dual to the homology region cross section and in \cite{Tamaoka_2019} the odd entanglement entropy.} (For recent progress in the context of the surface/state correspondence \cite{Miyaji:2015aa} see \cite{Bao:2018ac,Guo:2019aa}, and  in terms of bit threads see \cite{bao:aa}.) The EOP is defined as follows \cite{Terhal:2002aa}:
\begin{equation}\label{EOPdef}
E_{p}(\rho_{AB}) = \min_{\ket{\psi}_{AA'BB'}:\atop\psi_{AB}=\rho_{AB}}
S(\psi_{AA'})\,;
\end{equation}
here the minimization is over purifications $\ket{\psi}_{AA'BB'}$ of $\rho_{AB}$, where $A'$, $B'$ are auxiliary systems. In general this is a very difficult quantity to calculate, especially in QFTs (see e.g.  \cite{Bhattacharyya2018,Caputa_2019,Bhattacharyya:2019aa,Nguyen2018}). Despite this difficulty several general bounds are know for the EOP \cite{Bagchi:2015aa}:
\begin{outline}
\1 $\frac{1}{2}I(A:B) \leq E_{p}(A:B) \leq \min[S(A),S(B)]$
\1 $E_{p}(A:BC) \geq E_{p}(A:B)$
\1 $E_{p}(A:BC) \geq \frac{1}{2}I(A:B) + \frac{1}{2}I(A:C)$
\1 If $ABC$ is pure then $E_{p}(A:B) +E_{p}(A:C) \geq E_{p}(A:BC)$
\end{outline}
It has been shown via geometric proofs using surfaces that the minimal homology region cross section $E_{w}(A:B)$  too satisfies these conditions lending credibility to the conjecture 
\begin{equation}
E_{p}(A:B)=E_{w}(A:B)\,.
\end{equation}
Additionally, it was shown that the minimal homology region cross section satisfies the superadditivity condition
\begin{outline}
\1 $E_{w}(A\tilde{A}:B\tilde{B}) \geq E_{w}(A:B) + E_{w}(\tilde{A}:\tilde{B})$
\end{outline}
which $E_p$ does not always satisfy. If the conjecture holds, then this inequality would be a necessary criterion for a system to have a holographic dual; this is similar to how the holographic entanglement entropy always satisfies monogamy of mutual information (MMI), even though general quantum systems may violate MMI \cite{Hayden:2011ag,Headrick2014}.

In this paper we take a complementary approach to the holographic surface proposals: using the full power of convex duality and the max flow-min cut theorem (MFMC), we derive the dual convex flow program
\begin{equation}\label{floweopintro}
E_{w}(A:B) \coloneqq \max_{v^{\mu}} \int_{A} \sqrt{h} \,n_{\mu}v^{\mu} \quad \text{such that} \quad \nabla_{\mu}v^{\mu} = 0, \quad |v| \leq 1, \quad n_{\mu}v^{\mu} |_{m(AB)} = 0\,,
\end{equation}
which works by explicitly forcing any flow to remain inside the homology region $r(AB)$. This result was anticipated by \cite{Nguyen2018,Kudler-Flam:2018aa,Agon:2018aa,Bao:2018ac,Ghodrati:2019aa,Kudler-Flam:2019aa}. Using this construction we reprove the above bounds and interpret the result in terms of bit threads.

\begin{figure}[H]
\centering
\includegraphics[width=.5\textwidth]{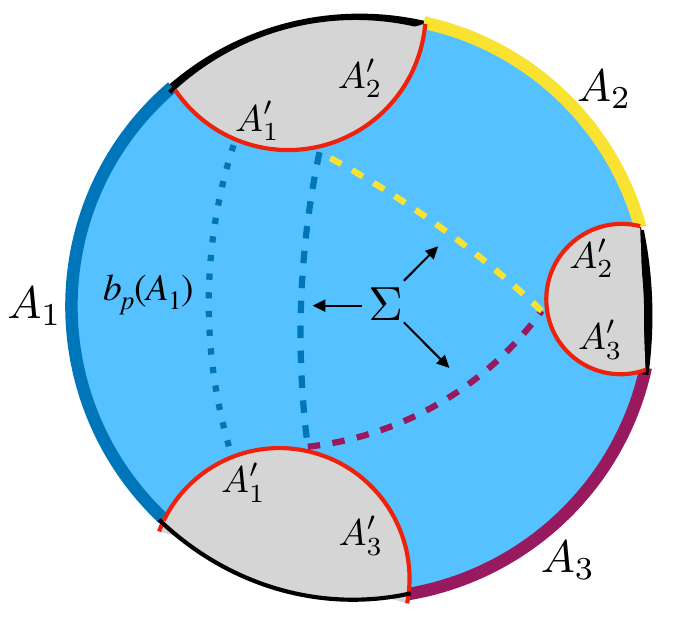}
\caption{\label{fig:multEOP}
The holographic multipartite EOP proposal of \cite{Umemoto2018}, illustrated for the case of three regions. The joint RT surface $\mathcal{O}\coloneqq m(\mathcal{A})$ is shown in red. This is partitioned into regions $A_i'$, and the total area of the corresponding minimal surfaces $m(A_iA_i')$ is minimized over partitions. The minimal surfaces $\Sigma(A_{i})$ are shown as colored dashed lines; their union is $\Sigma$. Also shown as a dotted line is the minimal surface $b_p(A_1)$ which computes the bipartite EOP $E_w(A_1)$.
}
\end{figure}

We then generalize our analysis to the case of multipartite EOP. Given a state $\rho_{\mathcal{A}}$ on a set  $\mathcal{A} = \{A_{i}\}$ of systems, this quantity is defined as
\begin{equation}\label{ITMEoPintro}
E_{p}(\mathcal{A}) = \min_{\ket{\psi}_{\mathcal{A}\mathcal{A}'}:\atop\psi_{\mathcal{A}}=\rho_{\mathcal{A}}} \sum_{i} S(\psi_{A_{i}A'_{i}})\,,
\end{equation}
where $\mathcal{A}' = \{A'_{i}\}$ is a set of auxiliary systems and the minimization is over purifications $\ket{\psi}_{\mathcal{A}\mathcal{A}'}$ of $\rho_{\mathcal{A}}$ \cite{Bao:2018aa,Umemoto2018}. The holographic dual to \eqref{ITMEoPintro} was conjectured by \cite{Umemoto2018} to be given as the area of a certain closed surface $\Sigma$ defined as follows. Consider all possible ways of partitioning the joint RT surface $\mathcal{O}\coloneqq m(\mathcal{A})$  into non-overlapping regions $A_i'$, and for each partition find the homologous minimal surface $m(A_iA_i')$; $\Sigma$ is their union, minimized over partitions (see fig.\ \ref{fig:multEOP}):
\begin{equation}\label{umemotoprop1}
E_{w}(\mathcal{A}) \coloneqq \min_{\{A_i'\}=\atop\text{partition of }\mathcal{O}} \sum_{i} \area(m(A_{i}A_{i}')) =  \sum_{i} \area(\Sigma(A_{i}))= \area(\Sigma)\,.
\end{equation}

To find the flow program dual to \eqref{umemotoprop1}, we must first make it into a convex program; this involves certain subtleties not present in the proof of the usual max flow-min cut theorem (see appendix \ref{sec:relax}). Modulo those issues, we find that the dual flow program involves a set $\mathcal{V}:=\{v_i\}$ of flows, one for each region $A_i$, along with a function $\alpha$ on $\mathcal{O}$ which defines a boundary condition for all of the flows:
\begin{equation}\label{meopintro}
\begin{split}
E_{w}(\mathcal{A}) = \max_{\mathcal{V},\alpha} \left(\sum_{i}\int_{A_{i}}\sqrt{h}\, n_{\mu}v^{\mu}_{i} + \int_{\mathcal{O}}\sqrt{h}\,\alpha \right) \\ \text{such that}\quad\nabla_{\mu}v_{i}^{\mu} = 0\,, \; |v_{i}| \leq 1\,, \;\left.n_{\mu}v^{\mu}_{i}\right|_{\mathcal{O}} = \alpha\,.
\end{split}
\end{equation}
We can understand this program heuristically as follows. If we set $\alpha=0$, then the flows are decoupled from each other; each $v_i$ is simply maximizing the flux from the corresponding region $A_i$ to the other regions $\cup_{j\neq i}A_j$. The optimal value is then just the sum of those bipartite cross-sections:
\begin{equation}\label{bpsat}
\int_{A_i}\sqrt h\,n_\mu v^\mu_i = E_w(A_i)
\end{equation}
where for brevity we write $E_w(A_i)$ for $E_w(A_i:\cup_{j\neq i}A_j)$. 
Turning on $\alpha$ allows each flow $v_i$ to have additional flux coming out of $\mathcal{O}$. However, the flows are now coupled to each other via the boundary condition $n_\mu v_j^\mu|_{\mathcal{O}}=\alpha$. This constraint forces the other flows $v_j$ ($j\neq i$) to also have this additional flux, limiting the amount of flux they can send to different regions. A maximal flow configuration is found by balancing the competition to maximize the flux of each flow. The duality guarantees this occurs exactly when each flow saturates on the corresponding portion of $\Sigma$.

It is clear from the definition \eqref{ITMEoPintro} that the multipartite EOP cannot be less than the sum of the respects bipartite EOPs:
\begin{equation}\label{multvsbi1}
E_p(\mathcal{A}) \ge \sum_iE_p(A_i)\,,
\end{equation}
This suggests that the residue
\begin{equation}\label{residue}
E_p(\mathcal{A}) - \sum_iE_p(A_i)\,,
\end{equation}
reflects the amount of entanglement that is ``truly multipartite'' among the parties of $\mathcal{A}$. The conjectured holographic duals obey the same inequality:
\begin{equation}\label{multvsbi}
E_w(\mathcal{A}) \ge \sum_iE_w(A_i)\,.
\end{equation}
This can be shown from the minimal surfaces \cite{Umemoto2018}, or from the flows using the above analysis, since the right-hand side is the maximum value of the objective subject to the additional constraint $\alpha=0$ (see Theorem \ref{T1} below). 
 
We will show that, within the set of all solutions (or maximizing configurations) of \eqref{meopintro}, there exists a subset of particularly nice solutions in which \eqref{bpsat} holds, i.e.\ the flux of $v_{i}$ emanating from $A_{i}$ saturates on the bipartite surface $b_p(A_{i})$. Hence the last term in the objective, the flux on $\mathcal{O}$, equals the residue of \eqref{multvsbi}:
\begin{equation}\label{Oflux}
\int_{\mathcal{O}}\sqrt h\,\alpha = E_w(\mathcal{A}) - \sum_iE_w(A_i)\,,
\end{equation}
suggesting again that this quantity reflect the truly multipartite entanglement.

The picture in terms of bit threads is that we have a different ``color'' of thread for each region $A_i$. These types of threads do not interact with each other in the bulk. They are only coupled via the constraint that, at any point on $\mathcal{O}$ where a thread of one color begins, a thread of each other color must also begin. In the ``nice'' solutions mentioned in the last paragraph, the $A_i$ threads connecting $A_i$ to $A_j$ ($j\neq i$) represent bipartite entanglement, while the ones connecting $\mathcal{O}$ to $A_j$ represent truly multipartite entanglement.

In section \ref{sec2} we articulate how to express the minimal homology region cross section in term of flows by utilizing relative homology and convex duality. We prove a simple theorem which describes how the addition of constraints affects convex optimization problems. Utilizing this theorem we reprove that all of the bipartite inequalities which the EOP satisfies are also satisfied by the minimal homology region cross section. In section \ref{sec3} we generalize to the multipartite case and provide a derivation using convex duality and proofs of several bounds. We find that with some simple restrictions we can ``gauge" to a class of thread configurations which have distinct bipartite and multipartite contributions. In \ref{sec4} we comment on some future directions of potential work. In appendix \ref{sec:mf} we review the definition and basic properties of the multiflow as well as give a flow-based proof of MMI. In appendix \ref{sec:relax} we discuss a subtlety in the convex relaxation of the holographic proposal \eqref{umemotoprop1}. In appendix \ref{sec:bulk} we dualize another interesting class of surfaces which allow for intersection in the bulk. Comparing to the holographic multipartite EOP, we demonstrate a common mechanism for how the manifold is able to provide the extra flux necessary to overcome geometric obstacles. Finally, in appendix \ref{sec:squashed} we provide evidence using an explicit bit thread construction of the multipartite information that the multipartite squashed entanglement will holographically always be saturated.

The methods of convex optimization, convex duality, and convex relaxation are used throughout this paper. Those unfamiliar with these methods may find it useful to review the relevant sections of \cite{Boyd_2004,Freedman2017,Headrick_2018}.

While this paper was in preparation, the paper \cite{Du:2019aa} appeared, which has some overlap with section \ref{sec2} of this paper.

\section{Bipartite entanglement of purification}\label{sec2}
\subsection{Flows and relative homology}
Starting with a Riemannian manifold $M$ which is dual to a state of a boundary CFT on $\partial M$, the entanglement entropy of a boundary region $A \subset \partial M$ is given by the RT formula \cite{Ryu_2006} as the minimal surface homologous to $A$
\begin{equation}
S(A) = \min_{m\sim A} \area(m).
\end{equation}
We define a flow $v^{\mu}$ as a norm-bounded divergenceless vector field on $M$. That is $\nabla_{\mu}v^{\mu} = 0$ and $|v^{\mu}| \leq 1$. Via an application of convex duality the RT formula can be expressed in terms of flows as the maximum flux through $A$ among all choices of flows \cite{Freedman2017}
\begin{equation}
S(A)= \max_{v^{\mu}} \int_{A} \sqrt{h} n_{\mu}v^{\mu}.
\end{equation}

It was shown in \cite{Headrick_2018} that we can make the following generalization of the standard max flow-min cut (MFMC) theorem for flows: If we define the minimal surface with respect to the relative homology $A\,\text{rel}\,\mathcal{O}$ of a boundary region $\mathcal{O}$ this has the effect of allowing the minimal surface to begin or end on $\mathcal{O}$, increasing the space of allowed surface. On the flow side this is equivalent to imposing the constraint $n_{\mu}v^{\mu}|_{\mathcal{O}} = 0$. This means no threads can end on or pass through $\mathcal{O}$.

\begin{figure}[H]
\centering
\includegraphics[width=.65\textwidth]{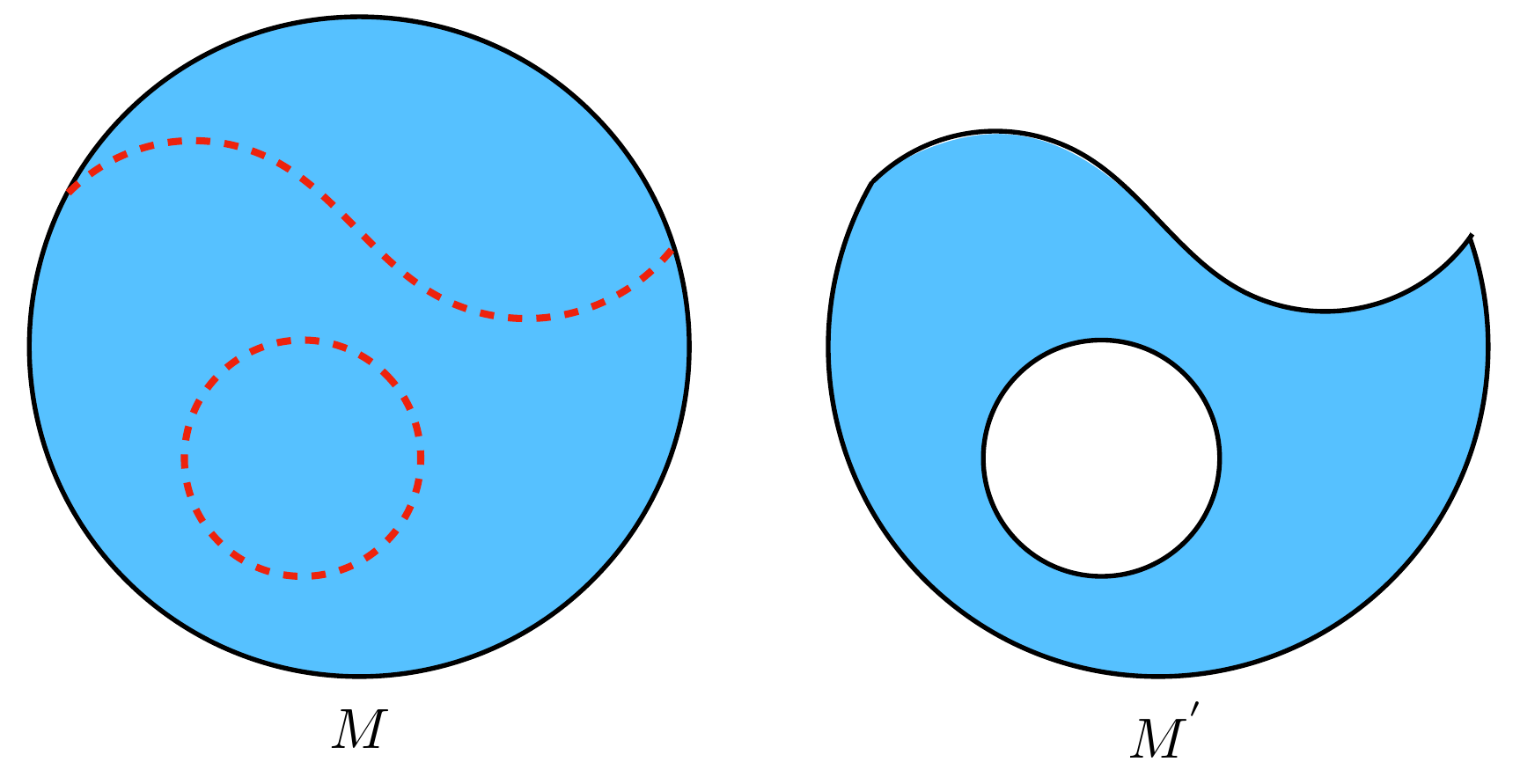}
\caption{\label{fig:shrinkingm}Shrinking $M$ to $M'$.}
\end{figure}

In order to make use of this we can further generalize to the case when $\mathcal{O}$ is instead an arbitrary codimension 1 surface whose boundary lies entirely on $\partial M$. In other words we no longer require that $\mathcal{O}$ be restricted to lie on the boundary $\partial M$. To prove this result we apply convex dualization to the max flow program
\begin{equation}
\max_{v^{\mu}} \int_{A} \sqrt{h} n_{\mu}v^{\mu} \quad \text{s.t.} \quad n_{\mu}v^{\mu}|_{\mathcal{O}}=0.
\end{equation}
First we construct the Lagrangian function on $M$
\begin{equation}
L(v^{\mu},\psi,\phi,\gamma) = \int_{A}\sqrt{h}n_{\mu}v^{\mu} + \int_{M}\sqrt{g}[-\psi\nabla_{\mu}v^{\mu} +\phi(1-|v^{\mu}|)] + \gamma \int_{\mathcal{O}}\sqrt{h}n_{\mu}v^{\mu}.
\end{equation}
To proceed we make the following observation: since the flow cannot penetrate the surface $\mathcal{O}$ and the minimal surface will always be located where the threads achieve their maximum density, the minimal surface corresponding to the dual solution can never be located in a region of the manifold contained entirely in $\mathcal{O}$ or between the boundary and $\mathcal{O}$. As a result without loss of generality we are free to instead consider the manifold $M'$ which is constructed by shrinking the bulk by removing all regions the threads cannot reach (see fig.\ \ref{fig:shrinkingm}). This ensures $\mathcal{O}$ will be a part of the boundary $\partial M'$.

Now that $\mathcal{O}$ is a part of the boundary we can make the constraint $n_{\mu}v^{\mu}|_{\mathcal{O}}=0$ implicit and integrate by parts
\begin{equation}
L(v^{\mu},\psi,\phi) = \int_{\partial M'\backslash \mathcal{O}}\sqrt{h}n_{\mu}v^{\mu}(\chi_{A}-\psi) + \int_{M'}\sqrt{g}[v^{\mu}\partial_{\mu}\psi +\phi(1-|v^{\mu}|)].
\end{equation}
Integrating out $v^{\mu}$ and minimizing we get the dual program
\begin{equation}
\min_{\psi,\phi} \int_{M'}\sqrt{g}\phi \quad s.t. \quad \psi |_{\partial M' \backslash \mathcal{O}} = \chi_{A}, \quad \phi \geq |\partial_{\mu}\psi| = \min_{m\sim A\,\text{rel}\,\mathcal{O}} \area(m) \; \text{on} \, M'.
\end{equation}
We may now view the surface as being located on $M$ instead of $M'$ and arrive at the MFMC result
\begin{equation}\label{eq:rhmfmc}
\max_{v^{\mu}} \int_{A} \sqrt{h} n_{\mu}v^{\mu} \quad \text{s.t.} \quad n_{\mu}v^{\mu}|_{\mathcal{O}}=0 \overset{\text{MFMC}}{\iff} \min_{m\sim A\,\text{rel}\,\mathcal{O}} \area(m) \; \text{on} \, M.
\end{equation}

\subsection{Flow formulation}

\begin{figure}[H]
\centering
\includegraphics[width=.4\textwidth]{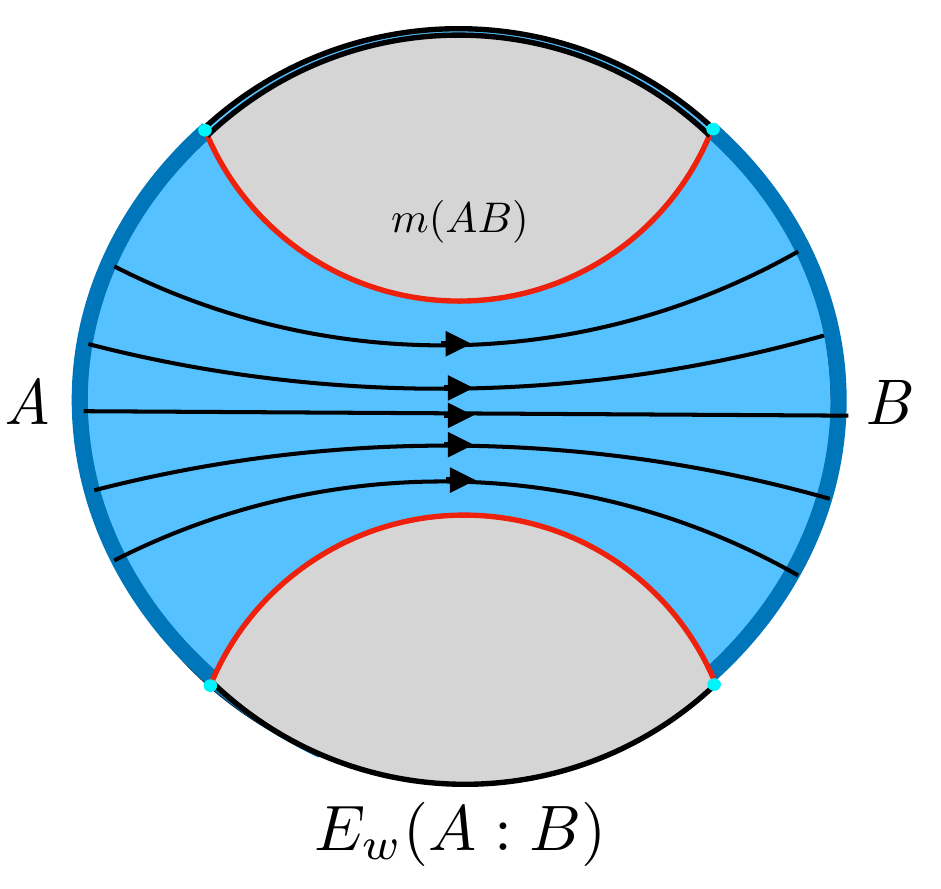}
\caption{\label{fig:eopboundary}A maximal flow of \eqref{floweop} whose flux calculates $E_{w}(A:B)$.}
\end{figure}

One can make use of \eqref{eq:rhmfmc} to define $E_{w}(A:B)$ in terms of flows: First, determine the homology region $r(AB)$ and its boundary in the bulk. From this construct the submanifold $M'$ with boundary $\partial M' = \partial r(AB) = A \cup B \cup m(AB)$. By using relative homology to force the thread configurations to remain in $M'$ we can then by \eqref{eq:rhmfmc} define $E_{w}(A:B)$  in terms of flows:
\begin{equation}\label{floweop}
E_{w}(A:B) \coloneqq \max_{v^{\mu}} \int_{A} \sqrt{h} n_{\mu}v^{\mu} \quad s.t. \quad n_{\mu}v^{\mu} |_{m(AB)} = 0 \iff \min_{b_{p}(A:B)\sim A \, \text{rel} \, m(AB)} \area(b_{p}(A:B))
\end{equation}
where $v^{\mu}$ is a flow (see fig.\ \ref{fig:eopboundary}). We stress that the explicit construction of $M'$ is not important since we can view the flow as being either on $M$ or $M'$. The existence of $M'$ provides a simple argument that all the standard results of convex optimization will carry over to this submanifold.

\subsection{Relaxation of the homology region}
In this subsection we will reformulate our flow program for the minimal homology region cross section by removing the need to refer to the boundary of the homology region. To provide some intuition consider the following: as before, split the boundary $\partial M$ into three regions $A,B,R:=(A\cup B)^{c}$ and calculate the maximum flux from $A$ to $B$ such that $n_{\mu}v^{\mu}|_{R}=0$. Such a maximum flow will not be contained to $r(AB)$. What we need is a way to force the threads to remain in $r(AB)$ so that the flux will correctly calculate $E_{w}(A:B)$.

\begin{figure}[hbt]
\centering
\includegraphics[width=.75\textwidth]{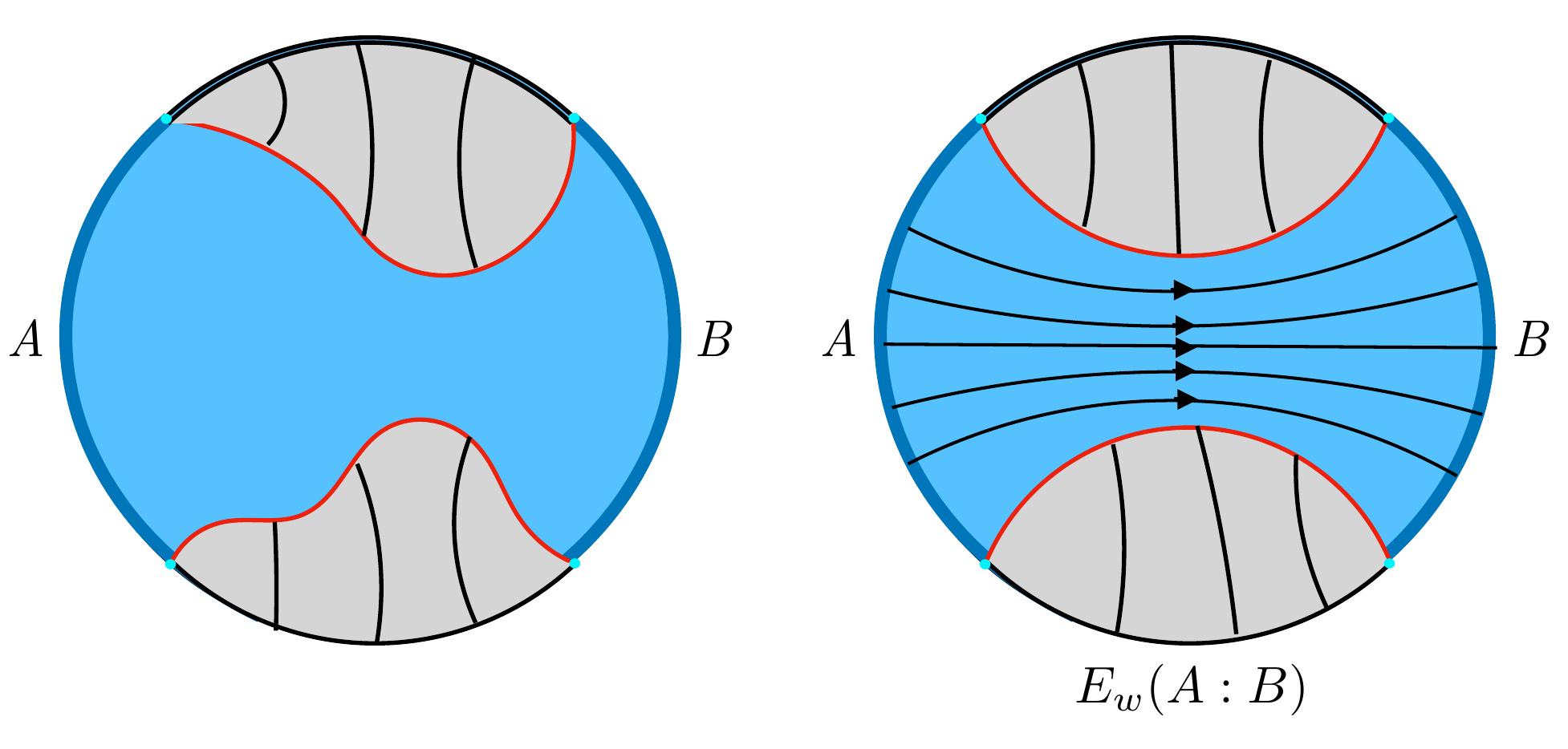}
\caption{\label{fig:generateew}Left: The flow $x^{\mu}$ acts to change the geometry $v^{\mu}$ can probe. Right: On the reduced manifold the maximum flux of $v^{\mu}$ gives the holographic EOP.}
\end{figure}

To do this consider an arbitrary partition of the manifold $M$ by a surface homologous to $R$, call it $\mathcal{O}$. For such a choice of $\mathcal{O}$ we calculate a maximum flow from $R$ to $\mathcal{O}$ and then minimize the maximum flux with respect to choices of partition. This has the effect of forcing $\mathcal{O}$ to extend far enough into the bulk such that the region between $R$ and $\mathcal{O}$ contains the minimal surface. Call this set of partitions $\mathcal{X}$.  By the divergencelessness condition all such elements of $\mathcal{X}$ will allow the same maximal flux. Next we calculate the maximum flow from $A$ to $B$ such that $n_{\mu}v^{\mu}|_{\mathcal{O}} =0$ where $\mathcal{O} \in \mathcal{X}$ . By construction the partition which gives the maximum space for these threads will be the one where $\mathcal{O}$ is the minimal surface homologous to $R$. $E_{w}(A:B)$ is the flux of this flow (see fig.\ \ref{fig:generateew}):
\begin{equation}
\mathcal{X} = \{\mathcal{O}\} \text{ s.t. } \min_{\mathcal{O}} \max_{x^{\mu}} \int_{R} x
\end{equation}
\begin{equation}
E_{w}(A:B) = \max_{v^{\mu},\mathcal{O} \in \mathcal{X}} \int_{A} v \text{ s.t. } n_{\mu}v^{\mu}|_{\mathcal{O}} =0.
\end{equation}

As a result of maximizing the flux on $\mathcal{O}$, the flows naturally find the minimal surface homologous to $\mathcal{O}$ which forms the boundary of the homology region. From this perspective the threads of $x^\mu$ act to change the geometry of the manifold by reducing the region that the threads of $v^\mu$ can occupy.
This means that there are a class of partitions of the manifold on which we will be able to reproduce $E_{w}(A:B)$. This could be anticipated by noting that the maximal flow only saturates at the minimal homology region cross section. There is freedom in how the threads are arranged which gives rise to a number of allowed partitions. While we will not in general make use of this additional freedom, it is appealing to see how the threads can give rise to the change in geometry. Since the homology region can always be canonically chosen using this procedure, we will in general make this choice leaving the generation of the manifold partition via threads implicit.
\subsection{Adding constraints}
The following bound will be used throughout the paper and is a key tool to proving many results. Therefore we state it as a theorem:
\begin{theorem} \label{T1} Let $P$ be a convex maximization program:
\begin{equation}
P: maximize \ f_{0}(y) \ subject \ to \ f_{i}(y) \leq0 \ \forall i \ and \ h_{j}(y) =0 \ \forall j
\end{equation}
with solution $y^{*} \in \mathcal{F}$ where $\mathcal{F}$ is the set of feasible points. Let $\tilde{P}$ be a second convex maximization program obtained by imposing an additional set of constraints $g_{m}(y) \leq0 \ \forall m \ and \ l_{n}(y) =0 \ \forall n$ on $P$:
\begin{equation}
\tilde{P}: maximize \ f_{0}(y) \ subject \ to \ f_{i}(y) \leq0 \ \forall i \ , \ h_{j}(y) =0 \ \forall j, \ g_{m}(y) \leq0 \ \forall_{m} \ and \ l_{n}(y) =0 \ \forall n
\end{equation}
and similarly define $\tilde{y}^{*} \in \tilde{\mathcal{F}}$. Then $f_{0}(y^{*}) \geq f_{0}(\tilde{y}^{*})$.
\end{theorem}
\begin{proof}
Suppose the opposite that is $f_{0}(y^{*}) < f_{0}(\tilde{y}^{*})$. By construction $\tilde{\mathcal{F}} \subset \mathcal{F}$ which means $\tilde{y}^{*} \in \mathcal{F}$ This leads to a contradiction since then $y^{*}$ is not maximal in $\mathcal{F}$.
\end{proof}

\begin{figure}[H]
\centering
\includegraphics[width=.5\textwidth]{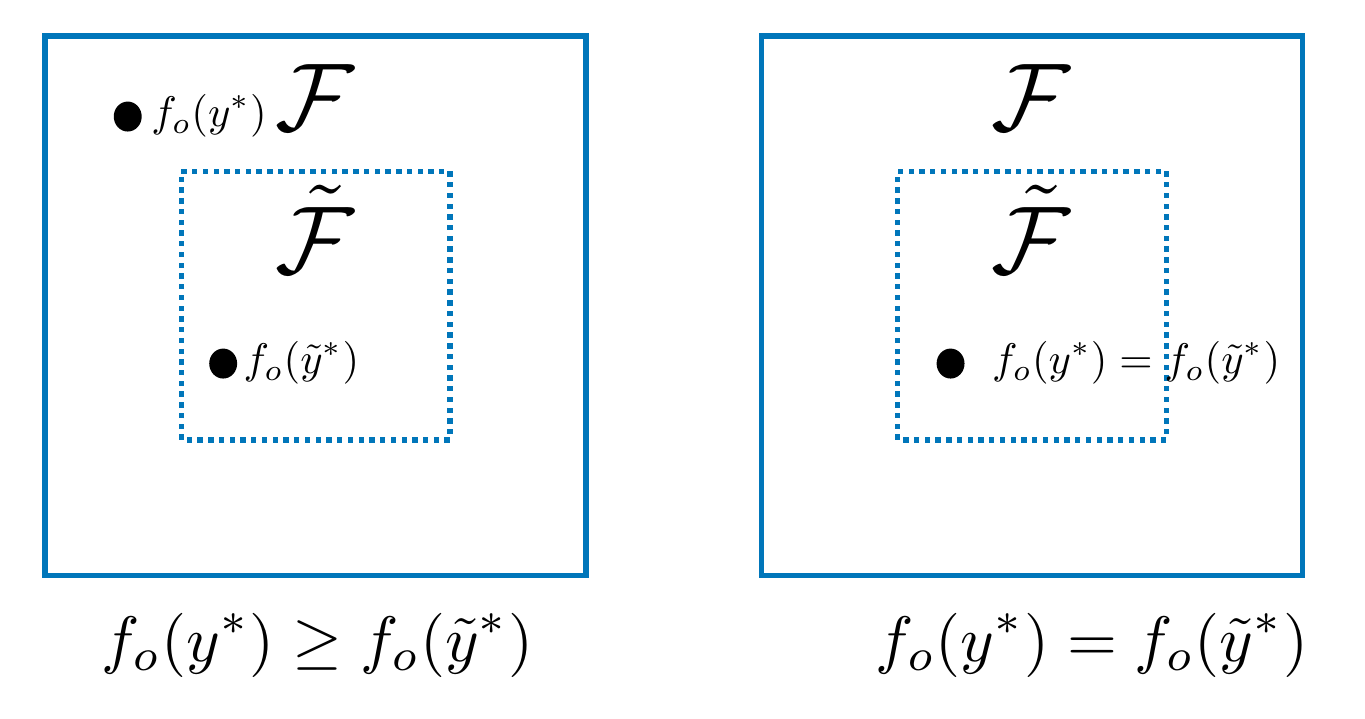}
\caption{\label{fig:constraint}The two possible cases: In each $f_{0}(y^{*}) \geq f_{0}(\tilde{y}^{*})$.}
\end{figure}

In other words there are two possible cases (see fig.\ \ref{fig:constraint}): If $y^{*} \in \tilde{\mathcal{F}}$ then necessarily $y^{*} = \tilde{y}^{*}$, but if $y^{*} \in \mathcal{F} \backslash \tilde{\mathcal{F}}$ then $f_{0}(y^{*})$ must be greater then every point in $\tilde{\mathcal{F}}$. If not then there would be a point in $\tilde{\mathcal{F}}$ which is greater and we would be back to the first case.

Theorem \ref{T1} states that whenever we have a convex maximization program and we add additional constraints, the original solution will always be an upper bound to the new one. Thus, by relating different programs in this way we can easily derive inequalities between them. This method allows for many of the inequalities involving EOP to be proven using our flow construction \eqref{floweop}.

\subsection{Conditional EOP}

Before continuing we wish to mention a straightforward generalization of the EOP. A proposal for holographic conditional EOP was defined in \cite{Bao2018,Bao:2018aa} as the minimum cross section of the bulk region $r(ABC)\backslash r(C)$. This is easily accounted for using bit threads by a modification of the boundary condition.
\begin{equation}
E_{w}(A:B|C) \coloneqq \max_{v^{\mu}} \int_{A} \sqrt{h} n_{\mu}v^{\mu} \quad s.t. \quad n_{\mu}v^{\mu} |_{\mathcal{O}[AB|C]} = 0
\end{equation}
where $\mathcal{O}[AB|C]\coloneqq \partial(r(ABC)\backslash r(C))\backslash A \cup B$. $E_{w}(A:B|C)$ satisfies a modified version of each of the inequalities that $ E_{w}(A:B)$ satisfies. The proofs remain essentially unchanged.
\subsection{Bounds}
We now proceed by deriving the inequalities that the EOP satisfies using our flow formulation for the holographic proposal. This is accomplished by noting that restricting to a smaller submanifold (i.e. decreasing the volume of the homology region) or adding topological obstructions, such as black holes, can only reduce the space threads can occupy. These restrictions are implemented by imposing additional constraints on the flow which increases the region of relative homology. Because of Theorem \ref{T1}, imposing a constraint can never increase the maximal flow which means the unconstrained problem is an upper bound. 

In contrast lower bounds can be found by explicit construction. If a particular thread configuration is always allowed when maximizing (in the feasible set of the convex program) then the true optimal value of the program will always be at least this value.
\subsubsection{ $\frac{1}{2}I(A:B) \leq E_{w}(A:B) \leq \min[S(A),S(B)]$}

\paragraph{Upper bound}
$S(A)$ and $E_{w}(A:B)$ only differ as flow optimization problems by the addition of the constraint $n_{\mu}v^{\mu}|_{m(AB)} =0$. Furthermore, we note that the flow which gives $E_{w}(A:B)$ is in the feasible set of flows of $S(A)$. Thus by Theorem \ref{T1} $S(A) \geq E_{w}(A:B)$. The same argument is true for $S(B)$ therefore $E_{w}(A:B) \leq \min[S(A),S(B)]$.

\paragraph{Lower bound}

\begin{figure}[H]
\centering
\includegraphics[width=.4\textwidth]{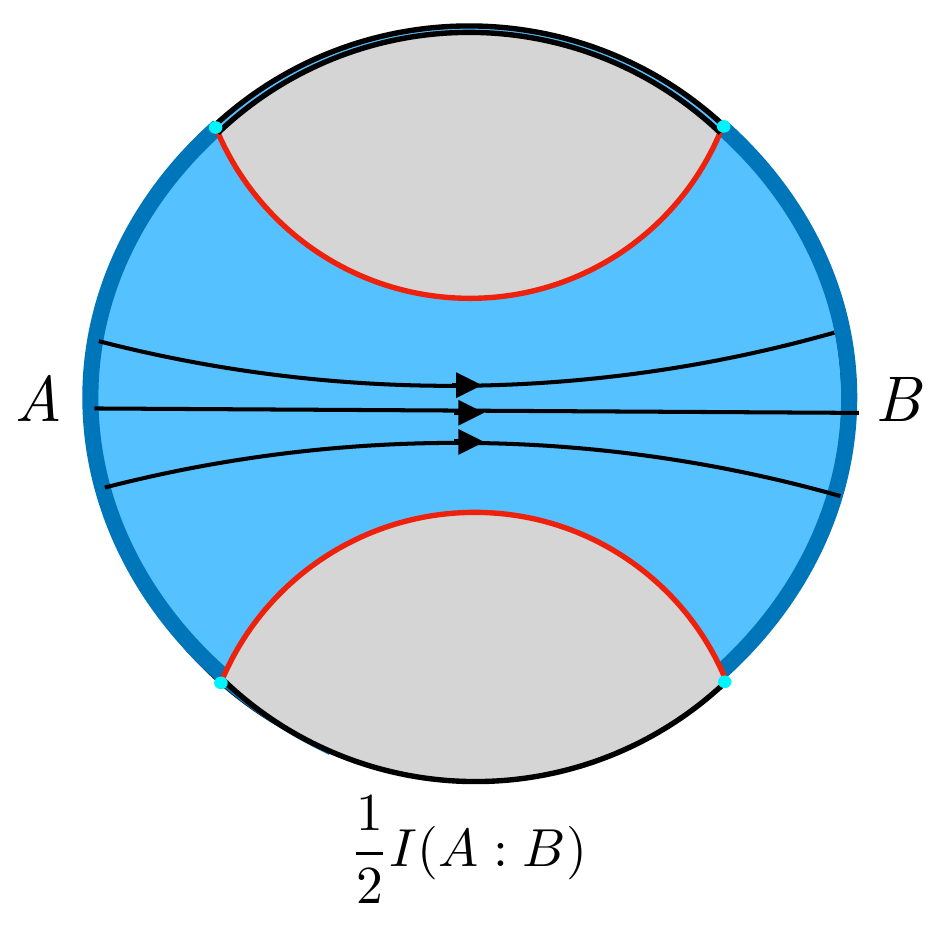}
\caption{\label{fig:mi}A flow which gives $\frac{1}{2}I(A:B)$ is suboptimal to the program \eqref{floweop}.}
\end{figure}

Following \cite{Freedman2017} , we can explicitly define the flow $\frac{1}{2}(v(A,B)-v(B,A))$ where $v(A,B)$ is the maximum flow through $A$ and $AB$. By construction $\frac{1}{2} I(A:B) = \int_{A}\frac{1}{2}(v(A,B)-v(B,A))$. This flow is non zero only in the homology region because both of the individual flows  $v(A,B)$ and $v(B,A)$ saturated on $m(AB)$. The result of taking their difference is therefore a flow which satisfies the condition $n_{\mu}v^{\mu}|_{m(AB)} =0$ and is in the feasible set of $E_{w}(A:B)$. Since it is an allowed flow when maximizing, $E_{w}(A:B)$ will always be at least this value therefore $\frac{1}{2}I(A:B) \leq E_{w}(A:B)$ (see fig.\ \ref{fig:mi}).

Together we have:
\begin{equation} \label{ineq1}
\frac{1}{2}I(A:B) \leq E_{w}(A:B) \leq \min[S(A),S(B)]
\end{equation}
\paragraph{Saturation}

\begin{figure}[hbt]
\centering
\includegraphics[width=.85\textwidth]{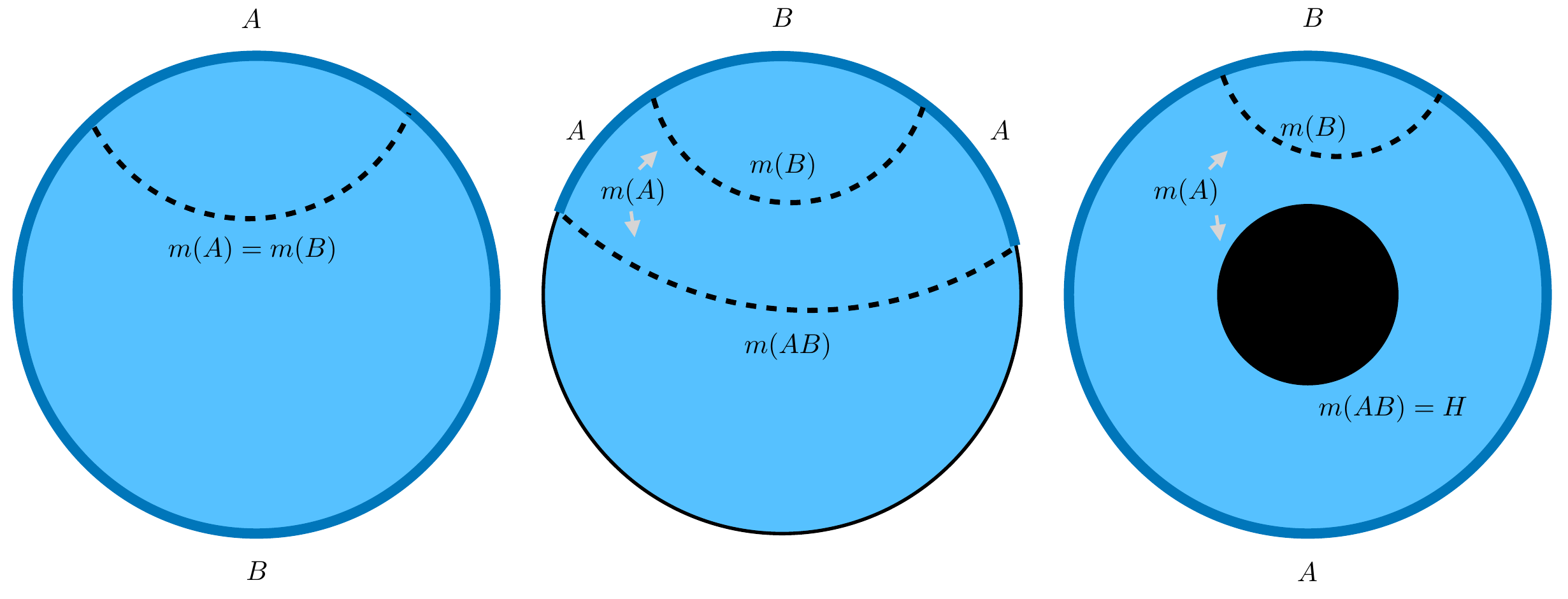}
\caption{\label{fig:AL}Examples of situations where AL is saturated.}
\end{figure}

The upper bound, $E_{w}(A:B) = \min[S(A),S(B)]$, will saturate if the Araki-Lieb (AL) inequality $S(A) \leq S(AB)+S(B)$ is saturated \cite{Nguyen2018} (see fig.\ \ref{fig:AL}). Without loss of generality let $S(B) \leq S(A)$. In the bulk AL will be saturated if and only if $m(A)=m(AB) \cup m(B)$ and $r(A)= r(AB) \backslash r(B)$ \cite{Headrick2014}.

In this case it is always possible find a flow that simultaneously maximizes the flux through $A,B,AB$ (which is an exceptional case of the nesting property of \cite{Freedman2017}). As a result, the part of this flow which goes to $B$ is the maximal flow and gives $E_{w}(A:B)$. That is when we impose $n_{\mu}v^{\mu}|_{m(AB)} =0$ none of the maximal flows which give $S(B)$ cross $m(AB)$ and thus imposing the constraint to get $E_{w}(A:B)$ cannot decrease the flow.

If we have a bipartite pure state both inequalities saturate as $S(AB) = 0$ and $m(A) = m(B)$. Trivially, $m(AB) = \emptyset$ and $ E_{w}(A:B) = S(A) =S(B)$. Therefore, $I(A:B) = 2S(A)$ and $\frac{1}{2}I(A:B) = E_{w}(A:B)$.

\subsubsection{$E_{w}(A:BC) \geq E_{w}(A:B)$}
We can define the two convex programs as:
\begin{equation}
\begin{split}
E_{w}(A:BC) &= \max_{v^{\mu}} \int_{A} \sqrt{h} n_{\mu}v^{\mu} \quad \text{s.t.} \quad n_{\mu}v^{\mu} |_{m(ABC)} = 0 \\
E_{w}(A:B) &= \max_{v^{\mu}} \int_{A} \sqrt{h} n_{\mu}v^{\mu} \quad \text{s.t.} \quad n_{\mu}v^{\mu} |_{m(ABC)} = 0, \; n_{\mu}v^{\mu} |_{m(AB)} = 0
\end{split}
\end{equation}

\begin{figure}[H]
\centering
\includegraphics[width=.4\textwidth]{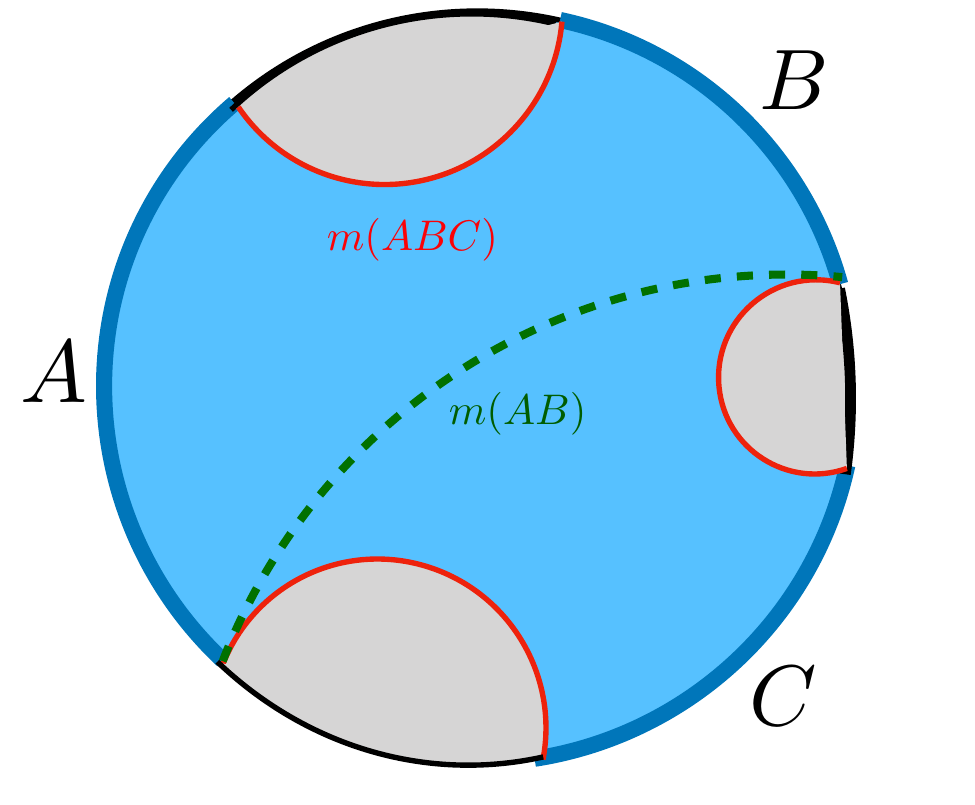}
\caption{\label{fig:3partybound}The reduction to $r(ABC)$. Imposing a second constraint reduces the manifold further to $r(AB)$.}
\end{figure}

Since we can get $E_{w}(A:B)$ by starting with $E_{w}(A:BC)$ and imposing an additional constraint and the maximal flow of $E_{w}(A:B)$ is an allowed flow of $E_{w}(A:BC)$, by Theorem \ref{T1} $E_{w}(A:B)$ can only be as big as $E_{w}(A:BC)$ (see fig.\ \ref{fig:3partybound})
\begin{equation} \label{ineq2}
E_{w}(A:BC) \geq E_{w}(A:B).
\end{equation}

\subsubsection{$E_{w}(A\tilde{A}:B\tilde{B}) \geq E_{w}(A:B) + E_{w}(\tilde{A}:\tilde{B})$}

\begin{figure}[hbt]
\centering
\includegraphics[width=.4\textwidth]{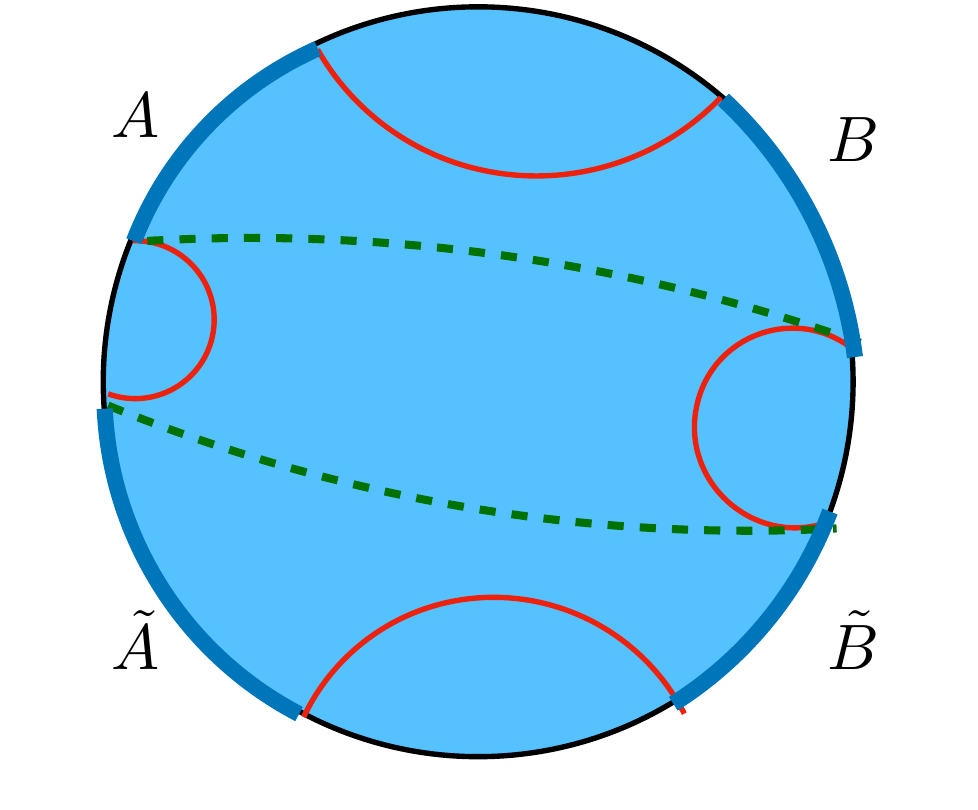}
\caption{\label{fig:4partybound}The green surfaces are the additional boundary constraints which must be imposed to separate $r(AB)$ and $r(\tilde{A}\tilde{B})$.}
\end{figure}

This superadditivity inequality always holds in the bulk, but will not hold on the boundary for a generic state. As a result it may be a criterion for determining whether a state has a holographic description similar to MMI \cite{Bhattacharyya2018}. The maximal flow which gives $E_{w}(A:B) + E_{w}(\tilde{A}:\tilde{B})$ is an allowed flow of the convex program of $E_{w}(A\tilde{A}:B\tilde{B})$, but $E_{w}(A:B) + E_{w}(\tilde{A}:\tilde{B})$ has the additional constraint $n_{\mu}v^{\mu} |_{\mathcal{O}'} =0$ which acts to separate $r(AB)$ and $r(\tilde{A}\tilde{B})$ (see fig.\ \ref{fig:4partybound}). As a result

\begin{equation} \label{ineq4}
E_{w}(A\tilde{A}:B\tilde{B}) \geq E_{w}(A:B) + E_{w}(\tilde{A}:\tilde{B})
\end{equation}
by Theorem \ref{T1}.
\paragraph{Saturation}
This inequality will saturate if and only if the boundary state is a product $\rho_{(A\tilde{A})(B\tilde{B})} = \rho_{AB} \otimes \rho_{\tilde{A}\tilde{B}}$. When this happens the mutual information $I(AB:\tilde{A}\tilde{B}) = 0$. In the bulk this means the two homology regions $r(AB)$ and $r(\tilde{A}\tilde{B})$ are disconnected \cite{Headrick2014} and implies there will be no maximal flow which has threads connecting $AB$ and $\tilde{A}\tilde{B}$. In this case no additional constraints are needed, $\mathcal{O}' = \emptyset$ and the two programs are equivalent sharing the same maximum.

\subsection*{Multiple Regions}
In order to proceed to cases where we deal with more than two boundary regions independently at a time we need to utilize some additional technology. Because now the flow will be connecting multiple regions we need a way to keep track of the threads extending between each region and ensure that the entire combination is still a valid flow. To this end we make use of the multiflow construction of \cite{Cui:2018aa} which accomplishes just this. Essentially, a multiflow is a collection of flows such that the linear combination of any flows in a region (referred to here as a subflow) is itself a flow. Those interested can find definitions and a flow proof of MMI in \ref{sec:mf}. 

In what follows, we will use multiflows to prove some remaining inequalities involving EOP. These inequalities involve more than two regions, but contain only quantities which are bipartite. As a consequence their proof requires only \eqref{ineq1} and MMI \eqref{MMI}. Since each of these inequalities has been proven using bit threads we cite the direct proofs using these properties.

\subsubsection{$E_{w}(A:BC) \geq \frac{1}{2}I(A:B) + \frac{1}{2}I(A:C)$}
By directly applying MMI \eqref{MMI} to \eqref{ineq1} we get
\begin{equation}
E_{w}(A:BC)\geq \frac{1}{2}I(A:BC) \geq \frac{1}{2}I(A:B) + \frac{1}{2}I(A:C).
\end{equation}
\subsubsection{If $ABC$ is pure $E_{w}(A:B) +E_{w}(A:C) \geq E_{w}(A:BC)$}
Consider a maximal multiflow $\mathcal{V}_{m}$ for a pure state of $\partial M = ABC$. By maximality we know the single interval subflows are saturated
\begin{equation}
S(A) = E_{w}(A:BC) = \int_{A} v_{A} \quad S(B) = \int_{B} v_{B} \quad S(C) = \int_{C} v_{C}.
\end{equation}
As a consequence of purity the two interval subflows are saturated as well and can be related to the single interval subflows
\begin{equation}
v_{AB} = -v_{C}, \quad v_{BC} = -v_{A}, \quad v_{AC} = -v_{B}.
\end{equation}
This implies the mutual information can be written simply as
\begin{equation}
I(A:B) = S(A) +S(B)-S(AB) = 2\int_{A}v_{AB}.
\end{equation}
Let $v'_{AB}$ be the maximal flow on the manifold $r(AB)$ whose flux gives $E_{w}(A:B)$ similarly define $v''_{AC}$ on $r(AC)$ for $E_{w}(A:C)$. Using the definition of the flows and the lower bound of \eqref{ineq1} we can rederive the inequality
\begin{equation}
\begin{split}
E_{w}(A:B) + E_{w}(A:C) = \int_{A}(v'_{AB}+v''_{AC}) \geq \frac{1}{2}(I(A:B) + I(A:C))= \\
\int_{A}(v_{AB}+v_{AC}) = \int_{A}v_{A} = E_{w}(A:BC).
\end{split}
\end{equation}

\section{Multipartite entanglement of purification}\label{sec3}

We will now generalize to the case of more than two boundary regions. The multipartite EOP was defined by \cite{Bao:2018aa,Umemoto2018} in the boundary CFT as
\begin{equation}\label{ITMEoP}
E_{p}(\mathcal{A}) = \min_{\ket{\psi}_{\mathcal{A}\mathcal{A}'}} \sum_{i} S(A_{i}A'_{i})
\end{equation}
where $\mathcal{A} = \{A_{i}\}$, $\mathcal{A}' = \{A'_{i}\}$, and the minimization is over purifications of $\rho_{\mathcal{A}}$.

\begin{figure}[H]
\centering
\includegraphics[width=.4\textwidth]{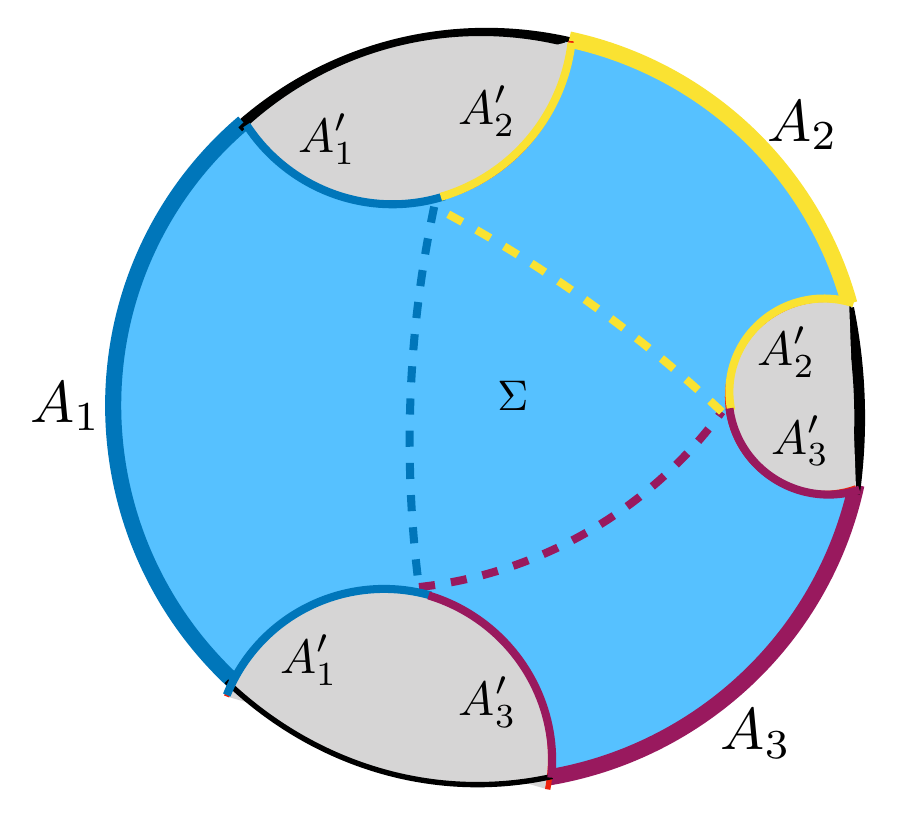}
\caption{\label{fig:meopp}The holographic multipartite EOP proposal of \cite{Umemoto2018}.}
\end{figure}

The holographic dual to \eqref{ITMEoP} was conjectured by \cite{Umemoto2018} to be given as the minimal closed surface among all possible partitions of the boundary region $\mathcal{O}$ into regions $A_{i}'$ each belonging to $A_{i}$\footnote{Historically, another surface proposal for the multipartite EOP was considered in \cite{Bao:2018aa} which allows for bulk intersection. Being an interesting class of of surfaces in the bulk we provide a dualization of these in \ref{sec:bulk}.} (see fig.\ \ref{fig:meopp})
\begin{equation}\label{umemotoprop}
E_{w}(\mathcal{A}) \coloneqq \min_{\mathcal{A}'} \sum_{i} \area(m(A_{i}A_{i}')) =\sum_{i} \area(\Sigma(A_{i}))  = \area(\Sigma)
\end{equation}
where $\Sigma \coloneqq \cup_{i}\Sigma(A_{i})$.

Before dualizing there is another geometric quantity we wish to mention for its importance in the interpretation of our construction. Consider a maximal multiflow which remains in the homology region $r(\mathcal{A})$:\footnote{In making this definition we have chosen to sum over all of the flows. This amounts to a different normalization than that of the bipartite case. This has been done for clarity and to remain notationally consistent with \cite{Bao:2018aa,Umemoto2018}.}
\begin{equation}
\area(b_{p}(\mathcal{A})) \coloneqq \max_{\mathcal{V}_{m}} \sum_{i,\; j}\int_{A_{i}} v_{ij} \text{ s.t. } n_{\mu}v_{ij}^{\mu}|_{m(\mathcal{A})} =0\,.
\end{equation}

\begin{figure}[H]
\centering
\includegraphics[width=.4\textwidth]{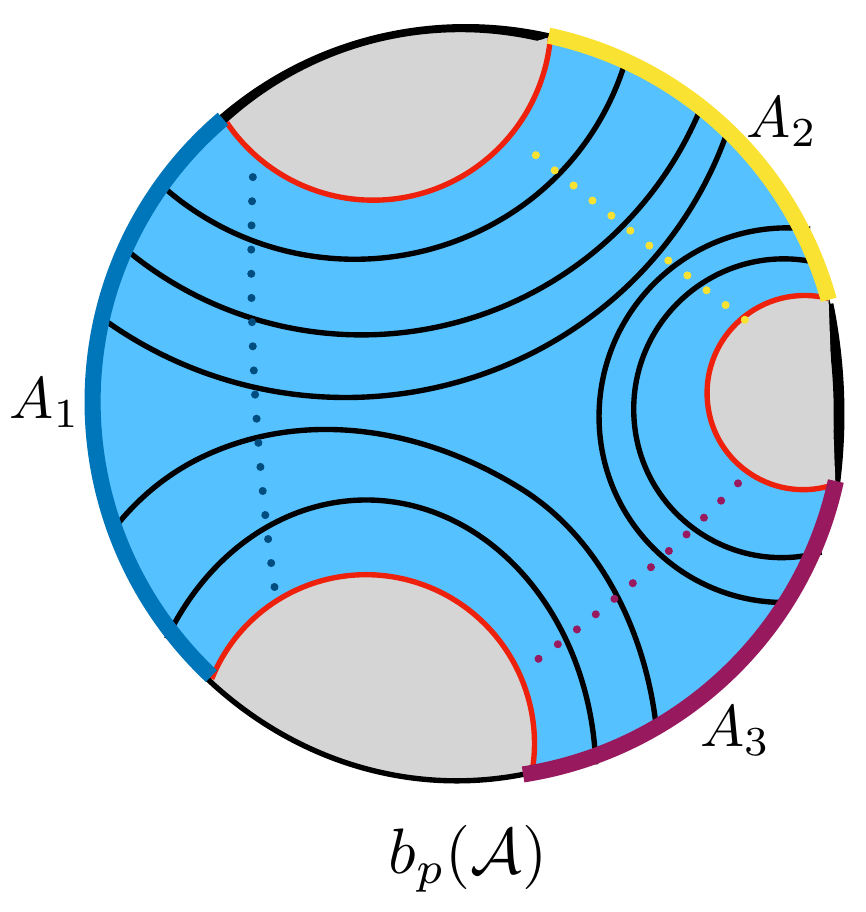}
\caption{\label{fig:bp}The sum of the area of the bipartite homology region cross sections for multiple regions is given by the flux of a maximal multiflow confined to $r(\mathcal{A})$. }
\end{figure}

\begin{figure}[H]
\centering
\includegraphics[width=.4\textwidth]{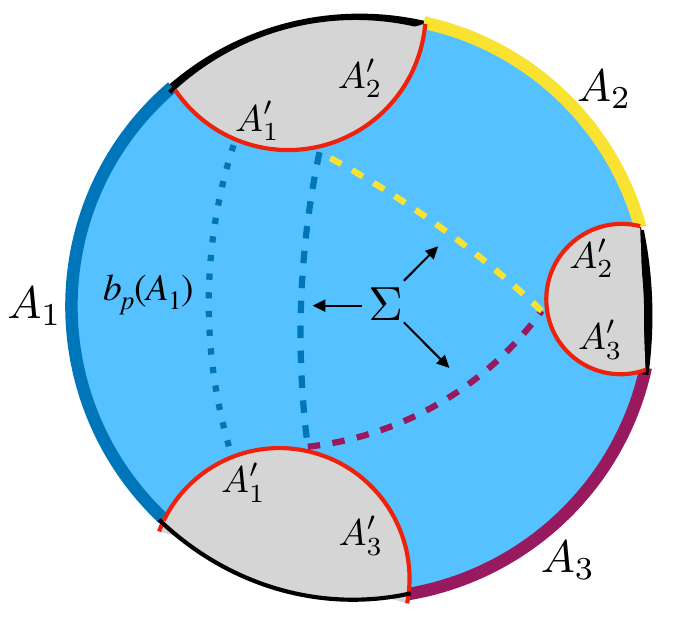}
\caption{\label{fig:sigmabpcolored2}The surface $b_{p}(\mathcal{A})$ is a geometric obstacle to the maximal flow. Additional flux is needed on $\mathcal{O}$ in order to get flows which saturate on $\Sigma$.}
\end{figure}

As a result of Theorem \ref{existmultiflow} the flux of this flow calculates precisely the sum of the bipartite homology region cross sections which is the minimal surface dual to the maximal multiflow on $r(\mathcal{A})$
\begin{equation}
\area(b_{p}(\mathcal{A}))= \sum_{i}E_{w}(A_{i} : \cup_{i \neq j}A_{j}).
\end{equation}
This surface is the natural geometric obstruction to flows in the bulk, but by definition is purely bipartite (see fig.\ \ref{fig:bp}). We will find that any formulation of multipartite entanglement must inherently find a way to circumvent this obstacle (see fig.\ \ref{fig:sigmabpcolored2}).\footnote{The surface $b_p(A_{i})$ is the minimal surface homologous to $A_{i}$ relative to $\mathcal{O}$. That is we can define a part of $\mathcal{O}$ called $A'_{i}$ which ``belongs" to $A_{i}$. Finding the surface $\Sigma(A_{i})$ is similar except for the additional condition that after splitting $\mathcal{O}$ among all of $\mathcal{A}$ there can be nothing left, everything must be assigned. This is more restrictive in that the resulting surface must be closed}

\subsection{Derivation}
\subsubsection{Convex and integer relaxation}
We will now derive our flow dual to \eqref{umemotoprop} using convex dualization. For clarity we will explicitly go through all the steps. In order to dualize we must first write this optimization problem as a convex program. Starting with \eqref{umemotoprop} we perform a convex relaxation by introducing scalar fields $\{\psi_{i}\}$  which smear the surface to form level sets in the bulk. For each $\psi_{i}$ we must impose boundary conditions based upon where the surface ends. We arrive at the relaxed program
\begin{equation}\label{mprelaxed}
\min_{\{\psi_{i}\}} \sum_{i}\int_{M}\sqrt{g} |\nabla_{\mu}\psi_{i}| \quad \text{s.t.} \quad \psi_{i}|_{A_{j}}=\delta_{ij}, \quad \sum_{i}\psi_{i}|_{\mathcal{O}} = 1 \quad \text{On $\mathcal{O}$: }\psi_{i}  \in \{0,1\}.
\end{equation}
Here the constraints on $\mathcal{O}$ force a division:  a piece of $\mathcal{O}$ must belong to only one $A_{i}$. The problem with the program as written is that it is an integer program and as a result is not convex. In order to make use of the tools of convex optimization it is necessary to perform an additional relaxation to the program. This is done by increasing the domain of the $\psi_{i}$'s to be over the reals\footnote{Such a relaxation of variables from being integer valued is very common in  the field of integer programming. Notably, if one has a linear objective and the constraints are totally unimodular then this relaxation will not introduce a so called integrality gap: the two program will have the same optimal configuration. This is the content of the Hoffman-Kruskal Theorem. A totally unimodular matrix is one which every square submatrix has determinant $-1,0,1$. Interestingly, our constraints though infinite in number have this property. Sadly, the absence of an integrality gap can not be guaranteed in the case of a convex objective.} resulting in a further relaxed program which is convex:
\begin{equation}\label{mprelaxedagain}
\min_{\{\psi_{i}\}} \sum_{i}\int_{M}\sqrt{g} |\nabla_{\mu}\psi_{i}| \quad \text{s.t.} \quad \left.\psi_{i}\right|_{A_{j}}=\delta_{ij}, \quad \left.\sum_{i}\psi_{i}\right|_{\mathcal{O}} = 1\,
\end{equation}

The consequence of this relaxation is an increase in the space of surfaces which are considered when minimizing. Namely, the relaxation now allows surface where regions of $\mathcal{O}$ are shared in a constrained way between the different regions $\mathcal{A}$. Generically, the two programs \eqref{mprelaxed} and \eqref{mprelaxedagain} are not the same and may have different minima. In fact we can find explicit examples where the two programs disagree. We do not know a complete set of sufficient and necessary conditions to guarantee agreement, but it is true that in many cases of interest we find that the original surface proposal and our relaxed program give the same result. This is essentially a consequence of convexity of the relevant surfaces which results from the negative curvature of the manifold. As we continue we will assume we are working in such cases when the two programs \eqref{mprelaxed} and \eqref{mprelaxedagain} are in fact equivalent and relegate further details and discussion of the integer relaxation to \ref{sec:relax}.

We wish to comment however that the relaxed program with its increased feasible space of surfaces represents an interesting class of geometric surfaces in the bulk. It is entirely possible that this generalization captures multipartite EOP as a separate proposal (which agrees with the current one in certain cases) or could represent another class of surfaces with a potential dual information theoretic interpretation.

\subsubsection{Dualization}
\label{sec:dualization}

We proceed with the dualization of \eqref{mprelaxedagain} by imposing each constraint with a Lagrange multiplier and find the Lagrangian:
\begin{equation}
\begin{split}
L(\{\psi\},\{\phi\},\{v^{\mu}\},\{\gamma\},\alpha) &= \sum_{i}\int_{M}\sqrt{g}\left[ |\phi_{\mu i}| +v^{\mu}_{i}(\phi_{\mu i}-\nabla_{\mu}\psi_{i})\right] \\
&+\int_{A_{i}}\sqrt{h}\left[\gamma_{ii}(\psi_{i}-1) + \sum_{i \neq j}\gamma_{ij}\psi_{j}\right] + \int_{\mathcal{O}} \sqrt{h}(-\alpha)(\psi_{i} -1).
\end{split}
\end{equation}
For convenience we define the covectors $\phi_{\mu i} = \nabla_{\mu}\psi_{i}$ the equivalence of which is imposed by the Lagrange multipliers $v^{\mu}_{i}$. These vector fields will be the dual variables which are our flows. We first dualize with respect to $\{\phi\}$ and get the standard bulk flow norm bound constraints:
\begin{equation} \phi_{\mu i}v^{\mu}_{i} + |\phi_{\mu i}| =0, \quad |v^{\mu}_{i}| \leq 1
\end{equation}
Utilizing these and performing integration by parts the Lagrangian becomes
\begin{equation}
\begin{split}
\sum_{i}\int_{M}\sqrt{g}\psi_{i}\nabla_{\mu}v^{\mu}_{i} +&\int_{A_{i}}\sqrt{h}\left[\gamma_{ii}(\psi_{i}-1) + \sum_{i \neq j}\gamma_{ij}\psi_{j} + \sum_{j} n_{\mu}v^{\mu}_{j}\psi_{j}\right] \\
+ &\int_{\mathcal{O}} \sqrt{h}\left[-\alpha(\psi_{i} -1) + \sum_{j} n_{\mu}v^{\mu}_{j}\psi_{j} \right].
\end{split}
\end{equation}
Finally we dualize with respect to the remaining primal variables $\{\psi\}$
\begin{equation}
\begin{split}
&\text{On $M$: } \nabla_{\mu}v^{\mu}_{i} =0\\
&\text{On $A_{i}$: } \gamma_{ii} + n_{\mu}v^{\mu}_{i} =0, \quad \gamma_{ij} + n_{\mu}v^{\mu}_{j} =0 \\
&\text{On $\mathcal{O}$: } -\alpha + n_{\mu}v^{\mu}_{i} =0. \\
\end{split}
\end{equation}
This gives the dual program which requires us to maximize over a set of flows $\mathcal{V}$ and the function $\alpha$ on $\mathcal{O}$:
\begin{equation}\label{meopv}
\begin{split}
E_{w}(\mathcal{A}) = \max_{\mathcal{V},\alpha} \left(\sum_{i}\int_{A_{i}}\sqrt{h} n_{\mu}v^{\mu}_{i} + \int_{\mathcal{O}}\sqrt{h}\alpha\right) \quad \text{s.t.} \\ \nabla_{\mu}v_{i}^{\mu} = 0, \; |v_{i}| \leq 1, \; \left.n_{\mu}v^{\mu}_{i}\right|_{\mathcal{O}} = \alpha\,.
\end{split}
\end{equation}
This flow program for the holographic multipartite EOP proposal is our main result.

\begin{figure}[H]
\begin{tabular}{ccc}
\includegraphics[width=.3\textwidth]{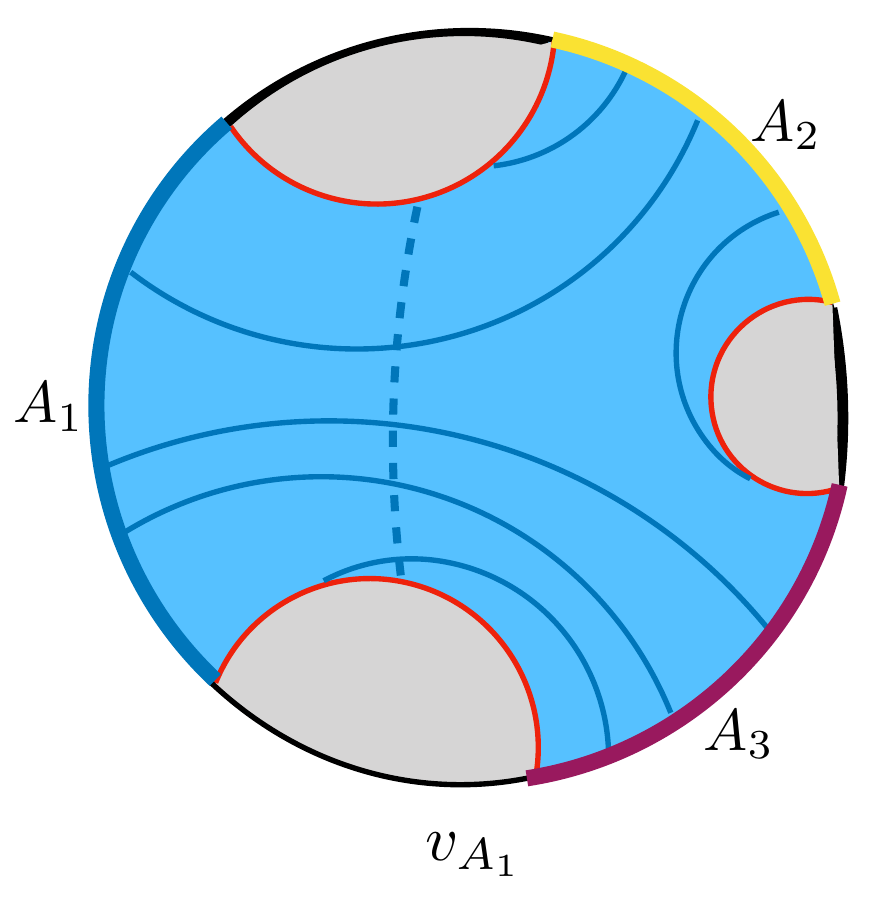} & \includegraphics[width=.3\textwidth]{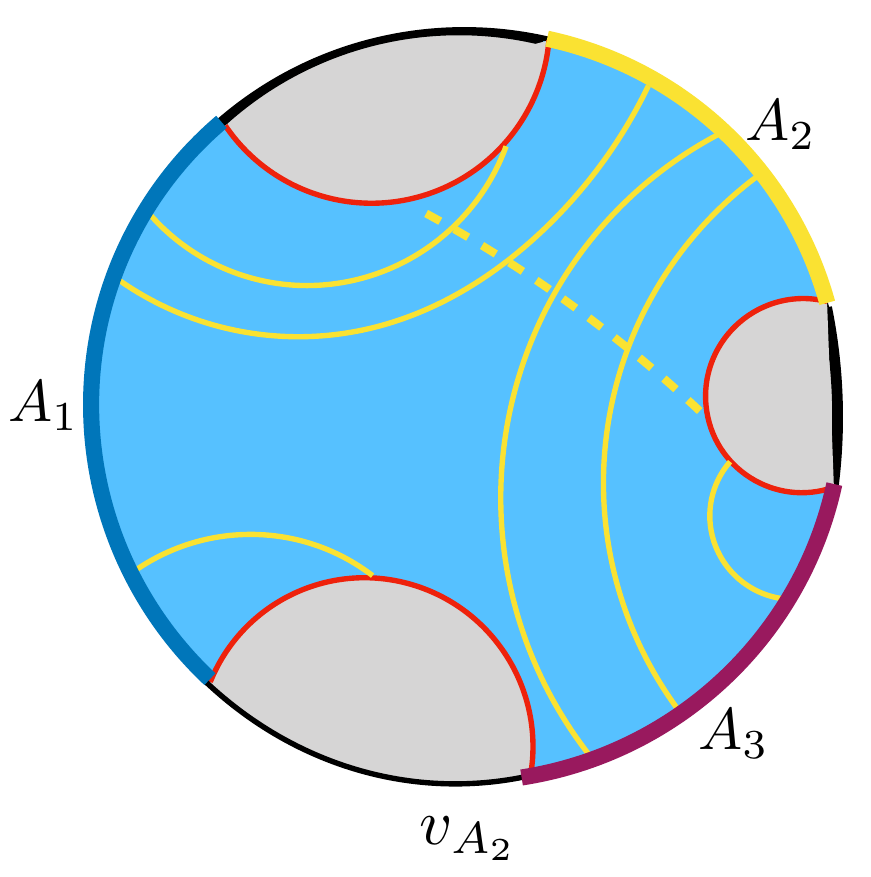} & \includegraphics[width=.3\textwidth]{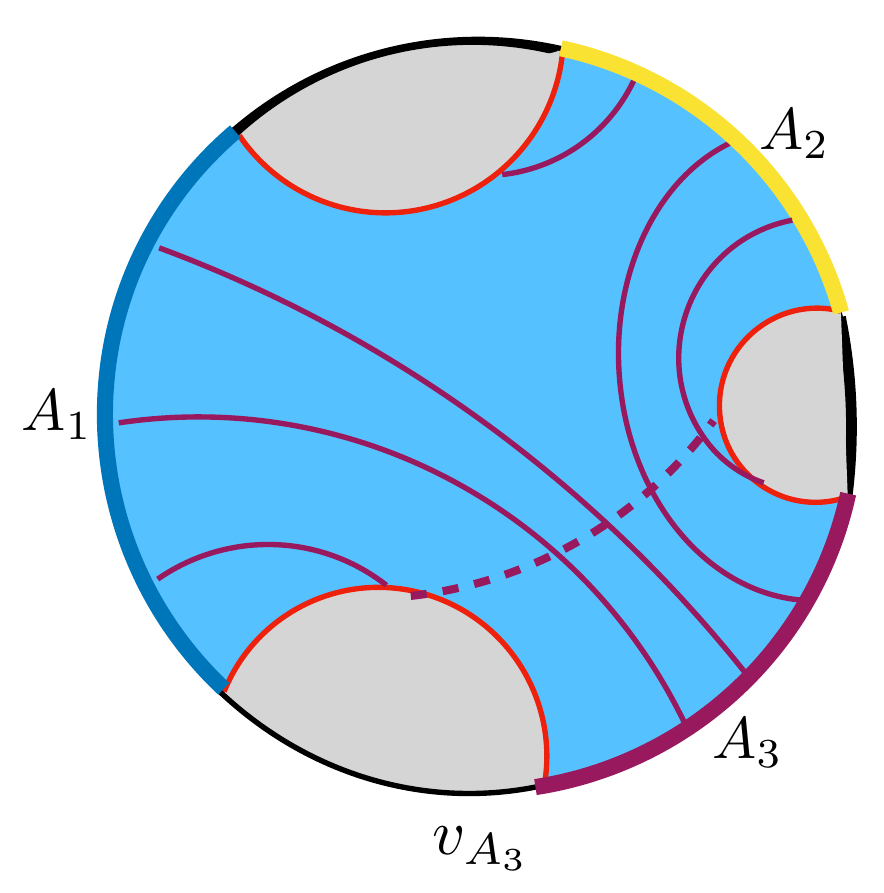}\\
\end{tabular}
\caption{\label{fig:meopexample}An example of a maximal flow configuration $\mathcal{V}$ to \eqref{meopv} for three regions. Note the local flux on $\mathcal{O}$ which we call $\alpha$ is the same for all three of the flows. In the example each surface $b_{p}(A_{i})$ can support three threads while each surface $\Sigma(A_{i})$ can support four.}
\end{figure}

The program \eqref{meopv} involves several flows, one for each boundary region $A_i$. These flows interact with each other only through the constraint $n_\mu v_i^\mu|_{\mathcal{O}}=\alpha$. The flow lines can be thought of as different-colored threads, or as living on different copies of the manifold that are glued together along $\mathcal{O}$.\footnote{This perspective has tantalizing visual overlap and possible connections to notions of wormhole geometries \cite{Dutta:aa,Bao:2018ac,Bao:2018ab}.}  As an example of an optimal configuration to provide intuition see fig.\ \ref{fig:meopexample}.

An important feature of \eqref{meopv} is that at no point were we required to define a partition of $\mathcal{O}.$\footnote{Since $\mathcal{O} = m(\mathcal{A})$ is the minimal surface homologous to $\mathcal{A}^{c}$ any threads on $\mathcal{O}$ can always nonuniquely be extended back to $\mathcal{A}^{c}$. In this sense one could view the flows as elements of the full manifold $M$. Moving forward we will not adopt this viewpoint.} That is the threads find the correct location of $\Sigma$ (where they saturate) naturally as a result of the maximization. This is in contrast to the cut side where we were required to explicitly minimize over both the area of surfaces and choice of partition. In a sense the lack of knowledge of the optimal partition is what makes this problem hard. Given the optimal partition $E_{w}(\mathcal{A})$ can be calculated trivially using standard flows which calculate entanglement entropies. This however hides the rich structure and multipartite nature which we will demonstrate.

For the reader familiar with the notion of a \emph{multiflow} \cite{Cui:2018aa}, we emphasize that these flows do \emph{not} define a multiflow since they do not necessarily obey a \emph{joint} norm bound. (See appendix \ref{sec:mf} for a review and discussion of multiflows.)

\subsection{Flow decomposition}
In order to interpret the maximal flows of our program \eqref{meopv}, we find it convenient to make some restrictions on the space of solutions. To do so we place additional constraints on the flow program which act to steer us towards a particular class of solutions. These solutions have the property of isolating the bipartite contribution from an additional term which is truly multipartite. This separation leads to a drastic simplification. 

Let $\mathcal{S} \subset \mathcal{F}$ be the space of maximal flow configurations, which is a subset of all of the feasible points of the convex program. We now define the subspace $\mathcal{S}'\subset \mathcal{S}$ with the following constraints:
\begin{outline}
\1 The flux on the boundary $\mathcal{O}$ is everywhere positive: $\alpha = n_{\mu}v^{\mu}_{i}|_{\mathcal{O}}\geq 0$.
\1 The total flux on the boundary $\mathcal{O}$ is as small as possible.
\1 The flow configuration contains no threads which contribute zero to the objective. This includes loops and threads which connect regions where the flux is not being counted.
\end{outline}

The last of these constraints is simply implemented for convenience and can always be done. It allows us to discuss ``simple" solutions which only contain threads which count towards the objective. Given a maximal thread configuration which contains such spurious threads we define a new solution by deleting them.

The result of this appropriate gauging of the solution is that the holographic multipartite EOP proposal has a natural decomposition into the number of threads connecting the regions $A_i$ to each other, which equals the total bipartite homology region cross section $\sum\area(b_p(A_i))$, and an additional term counting the number of threads on $\mathcal{O}$, which is fundamentally multipartite. This leads naturally to the conjecture that the actual multipartite EOP may exhibit a similar decomposition:
\begin{equation}
E_{p}(\mathcal{A}) = \sum_{i}E_{p}(A_{i} : \cup_{i \neq j}A_{j}) + M(\mathcal{A})
\end{equation}
where $M(\mathcal{A})$ is the multipartite contribution.

Our task is to implement each of the above constraints in the convex maximization program and understand how they respond to dualization. This will allow us to define a modified dual which implements the restriction to $\mathcal{S}'$. We will argue that each of these constraints does not affect the optimal value of our program.

\subsubsection{Forcing $\alpha\ge0$}

The first restriction we would like to add to the flow program \eqref{meopv} is $\alpha\ge0$:
\begin{equation}\label{meopv2}
\begin{split}
\tilde E_w(\mathcal{A})=\max_{\mathcal{V},\alpha} \left(\sum_{i}\int_{A_{i}}\sqrt{h}\, n_{\mu}v^{\mu}_{i} + \int_{\mathcal{O}}\sqrt{h}\,\alpha\right) \quad \text{s.t.} \\ \nabla_{\mu}v_{i}^{\mu} = 0\,, \; |v_{i}| \leq 1\,, \; \left.n_{\mu}v^{\mu}_{i}\right|_{\mathcal{O}} = \alpha\,,\;\alpha\ge0\,.
\end{split}
\end{equation}
In the dualization presented in subsection \ref{sec:dualization} above, $\alpha$ arose as the Lagrange multiplier for the equality constraint $\sum\psi_i|_{\mathcal{O}}=1$. According to the rules of convex dualization, an \emph{inequality} constraint is enforced by a \emph{non-negative} Lagrange multiplier, so changing this to an inequality constraint $\sum\psi_i|_{\mathcal{O}}\ge1$ simply leads to the additional inequality constraint $\alpha\ge0$ on the dual, which is just what we want. Thus the dual of \eqref{meopv2} is as follows:
\begin{equation}\label{mprelaxedagain2}
\tilde E_w(\mathcal{A})=\min_{\{\psi_{i}\}} \sum_{i}\int_{M}\sqrt{g} |\nabla_{\mu}\psi_{i}| \quad \text{s.t.} \quad \left.\psi_{i}\right|_{A_{j}}=\delta_{ij}, \quad \left.\sum_{i}\psi_{i}\right|_{\mathcal{O}} \ge 1\,.
\end{equation}

We will now argue that these two programs are equivalent to \eqref{meopv} and \eqref{mprelaxedagain}, i.e.\ that the new constraint $\alpha\ge0$ and the new relaxation $\sum\psi_i|_{\mathcal{O}}\ge1$ have no affect on the optimal value of these programs, and thus $\tilde E_w(\mathcal{A})=E_w(\mathcal{A})$.

We have two arguments for this equivalence. First, as argued in appendix \ref{sec:relax}, we don't expect that in the solution to \eqref{mprelaxedagain} the regions where the different $\psi_i$s are non-zero will overlap (at least in geometries relevant for AdS/CFT dualities), and the same argument applies to \eqref{mprelaxedagain2}. In this case, we can show that the solution must have $\psi_i\le1$ everywhere: if $\psi_i>1$ somewhere, setting it to 1 reduces the value of the objective while maintaining the boundary conditions. Therefore $\sum\psi_i\le1$ everywhere, and in particular on $\mathcal{O}$ the sum equals 1.

The second argument applies complementary slackness to the new inequality constraint. This says that, if at a given point $x\in\mathcal{O}$, there exists an optimal flow configuration with $\alpha(x)>0$, then for \emph{all} optimal configurations of \eqref{mprelaxedagain2}, $\sum\psi_i(x)=1$. Given the well-known non-uniqueness of max flows, it would be surprising if there was no way to have a non-zero flux at some given point on $\mathcal{O}$.

\subsubsection{Minimizing flux on $\mathcal{O}$}

In the program \eqref{meopv2}, the flow $v_i$ has non-negative flux on $\mathcal{O}$. Therefore, in any feasible configuration, its flux on $A_i$ is bounded above by the bipartite homology-region cross section:
\begin{equation}\label{s'upper}
\int_{A_{i}}\sqrt h\,n_\mu v^\mu_{i} \leq \area(b_p(A_{i}))\,.
\end{equation}
We will now show that the bounds \eqref{s'upper} are collectively tight: there exist optimal configurations in which they are simultaneously saturated for all $i$.

In the space $\mathcal{S}$ of solutions of \eqref{meopv2}, the different terms in the objective do not necessarily take fixed values: a given term can take different values on different solutions, although of course the total value of the objective is the same in all the solutions. We are interested in solutions in which the first term is as large as possible and the second term (the flux on $\mathcal{O}$) as small as possible. To steer the solution in that direction, we add to the objective a term $\epsilon\sum\int_{A_i}v_i$, where $\epsilon$ is small and positive, yielding the following
modified program:
\begin{equation}\label{meopv3}
\begin{split}
\tilde E^\epsilon_w(\mathcal{A})=\max_{\mathcal{V},\alpha} \left((1+\epsilon)\sum_{i}\int_{A_{i}}\sqrt{h}\, n_{\mu}v^{\mu}_{i} + \int_{\mathcal{O}}\sqrt{h}\,\alpha\right) \quad \text{s.t.} \\ \nabla_{\mu}v_{i}^{\mu} = 0\,, \; |v_{i}| \leq 1\,, \; \left.n_{\mu}v^{\mu}_{i}\right|_{\mathcal{O}} = \alpha\,,\;\alpha\ge0\,.
\end{split}
\end{equation}
In general, for concave functions $f$ and $g$,
\begin{equation}
\max(f+\epsilon g) = \max f +\epsilon\max_{\mathcal{S}}g
\end{equation}
(to first order in $\epsilon$), where $\mathcal{S}$ is the set of maximizers of $f$. In this case, we have
\begin{equation}
\tilde E^\epsilon_w(\mathcal{A})=\tilde E_w(\mathcal{A})+\epsilon\max_{\mathcal{S}}\sum_{i}\int_{A_{i}}\sqrt{h}\, n_{\mu}v^{\mu}_{i}\,.
\end{equation}

We will show that the optimal value is bounded below as follows:
\begin{equation}\label{s'lower}
\tilde E_w^\epsilon(\mathcal{A}) \ge \tilde E_w(\mathcal{A}) + \epsilon\sum_i\area(b_p(A_i))\,,
\end{equation}
so
\begin{equation}\label{s'lower2}
\max_{\mathcal{S}}\sum_{i}\int_{A_{i}}\sqrt{h}\, n_{\mu}v^{\mu}_{i}\ge\sum_i\area(b_p(A_i))\,.
\end{equation}
Given the inequalities \eqref{s'upper}, this can only hold if there exist solutions to \eqref{meopv2} in which all of those inequalities are saturated (and thus \eqref{s'lower2} is also saturated).

To show \eqref{s'lower}, we first note that, by a straightforward dualization, \eqref{meopv3} is equivalent to a version of \eqref{mprelaxedagain2} with a modified boundary condition:
\begin{equation}\label{restrictedcut}
\tilde E^\epsilon_w(\mathcal{A})= \min_{\{\psi_{i}\}} \sum_{i}\int_{M}\sqrt{g}\, |\nabla_{\mu}\psi_{i}| \quad \text{s.t.} \quad \left.\psi_{i}\right|_{A_{j}}=(1+\epsilon)\delta_{ij}\,, \quad \left.\sum_{i}\psi_{i}\right|_{\mathcal{O}} \geq 1\,.
\end{equation}
Consider now  a feasible configuration of \eqref{restrictedcut}. We can turn it into a feasible configuration for \eqref{mprelaxedagain2} as follows: Wherever $\psi_i<0$, set it to 0, and wherever $\psi_i>1$, set it to 1. Since it is a feasible configuration for \eqref{mprelaxedagain2}, the objective has value no less than $\tilde E_w(\mathcal{A})$. This operation reduces the objective because it removes level sets. (Recall that, by the co-area formula, the term $\int|\nabla\psi_i|$ in the objective equals the area of the level set $\{x:\psi_i(x)=p\}$, integrated over $p$.) In particular, the level sets for $1 < \psi_i < 1+\epsilon$ are homologous to $A_i$ relative to $\mathcal{O}$, and therefore have area at least $b_p(A_i)$. So the reduction in the objective is at least
\begin{equation}
\epsilon\sum_i\area(b_p(A_i))\,,
\end{equation}
which establishes the inequality \eqref{s'lower}.

We will denote the set of solutions to \eqref{meopv3} that saturate \eqref{s'upper} $\mathcal{S}'$. These decompose the holographic multipartite EOP into bipartite and ``truly multipartite'' contributions.\footnote{Because of the non-uniqueness of thread configurations there is a freedom to convert some of the bipartite threads into threads on $\mathcal{O}$ which we would here refer to as multipartite contributions. Such thread configurations are not in $\mathcal{S'}$. Taking this idea seriously there is some freedom due to space in the geometry to convert between bipartite and multipartite contributions. This freedom could have an interesting information theoretic interpretation in the boundary CFT, but we leave such exploration to future work. Importantly, by restricting to $\mathcal{S}'$ we are maximizing the number of bipartite threads. It is in this sense that the remainder of the flux, which must be on $\mathcal{O}$, \emph{must} be ``truly multipartite".} From now on we will assume that all of the solutions of \eqref{meopv} we consider are in $\mathcal{S}'$ unless otherwise stated.

\subsection{$\mathcal{V}$ flows}

\begin{figure}[H]
\centering
\includegraphics[width=.5\textwidth]{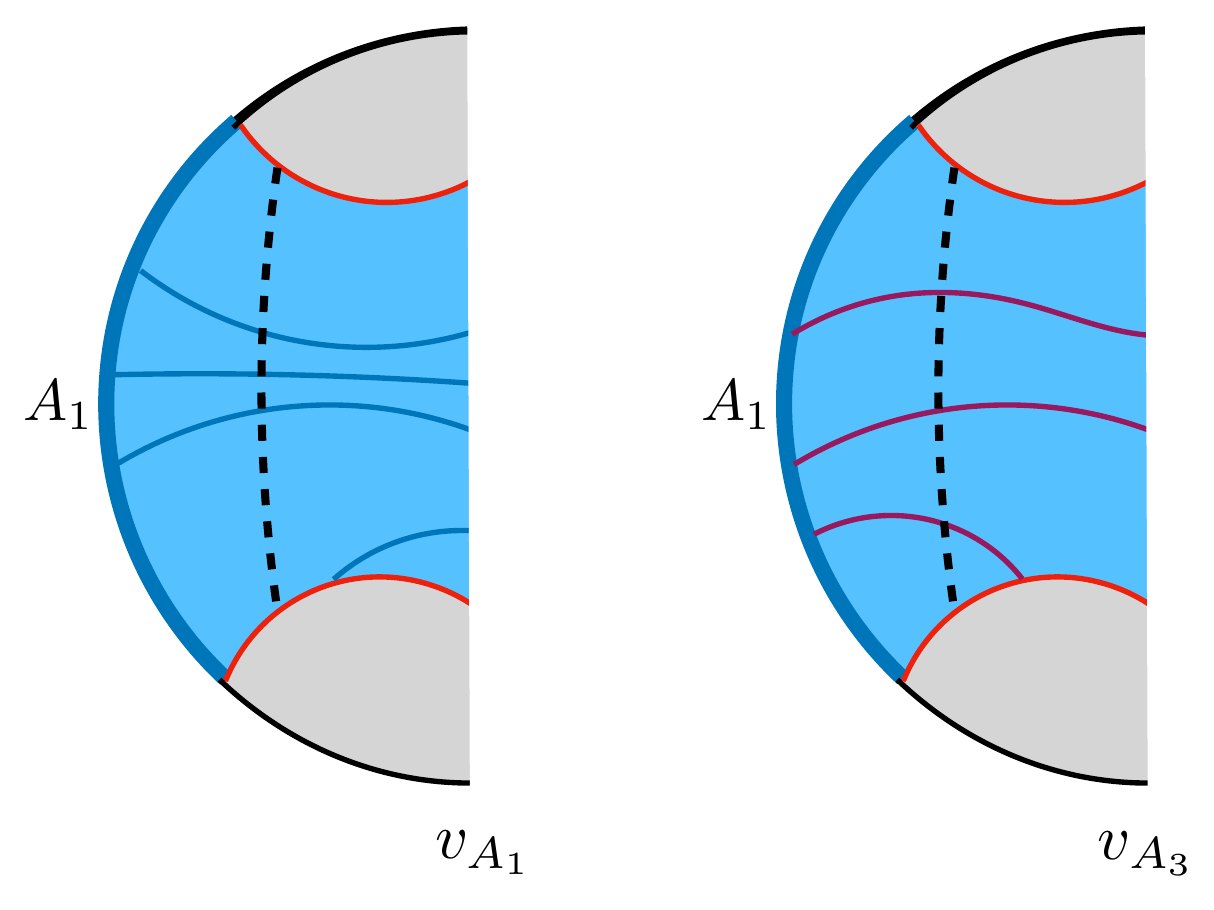}
\caption{\label{fig:voneregion}A portion of the manifold from $A_{1}$. In the example $b_{p}(A_{1})$ can support 3 threads. Left: The flow $v_{A_{1}}$. Right: The incoming flow from another region is limited by the required flux on $\mathcal{O}$.}
\end{figure}

To understand the behavior of $\mathcal{V}$ consider one of the flows $v_{A_{i}}$. It is always possible to send enough flow from $A_{i}$ through $b_{p}(A_{i})$ such that the flow saturates on $b_p(A_{i})$. This is clear from the multiflow construction of $b_{p}(\mathcal{A})$. For the flow to also saturate on $\Sigma(A_{i})$ there must be additional flux from the portion of $\mathcal{O}$ between $b_p(A_{i})$ and $\Sigma(A_{i})$ of the amount $\area(\Sigma(A_{i})) - \area(b_p(A_{i}))$. However, because of the condition $\alpha =n_{\mu}v^{\mu}_{i}|_{\mathcal{O}} = n_{\mu}v^{\mu}_{j}|_{\mathcal{O}}$ this means the same amount of flux must be present for every other flow. These threads must go back to one of the boundary regions, for concreteness say $A_{i}$, and as a result cross $b_p(A_{i})$. Note that doing so limits the amount of flux that the other flows can send back to $A_{i}$ (see fig\ \ref{fig:voneregion}). It is through this mechanism of balancing the amount of flux that leaves $\mathcal{O}$ conditioned on maximizing the flux each region can send that the flows find the correct location of $\Sigma$ and value of $E_{w}(\mathcal{A})$. For illustration, a numerical example is presented in fig.\ \ref{fig:flowexample}.

\begin{figure}[H]
\centering
\begin{tabular}{cc}
\includegraphics[width=.4\textwidth]{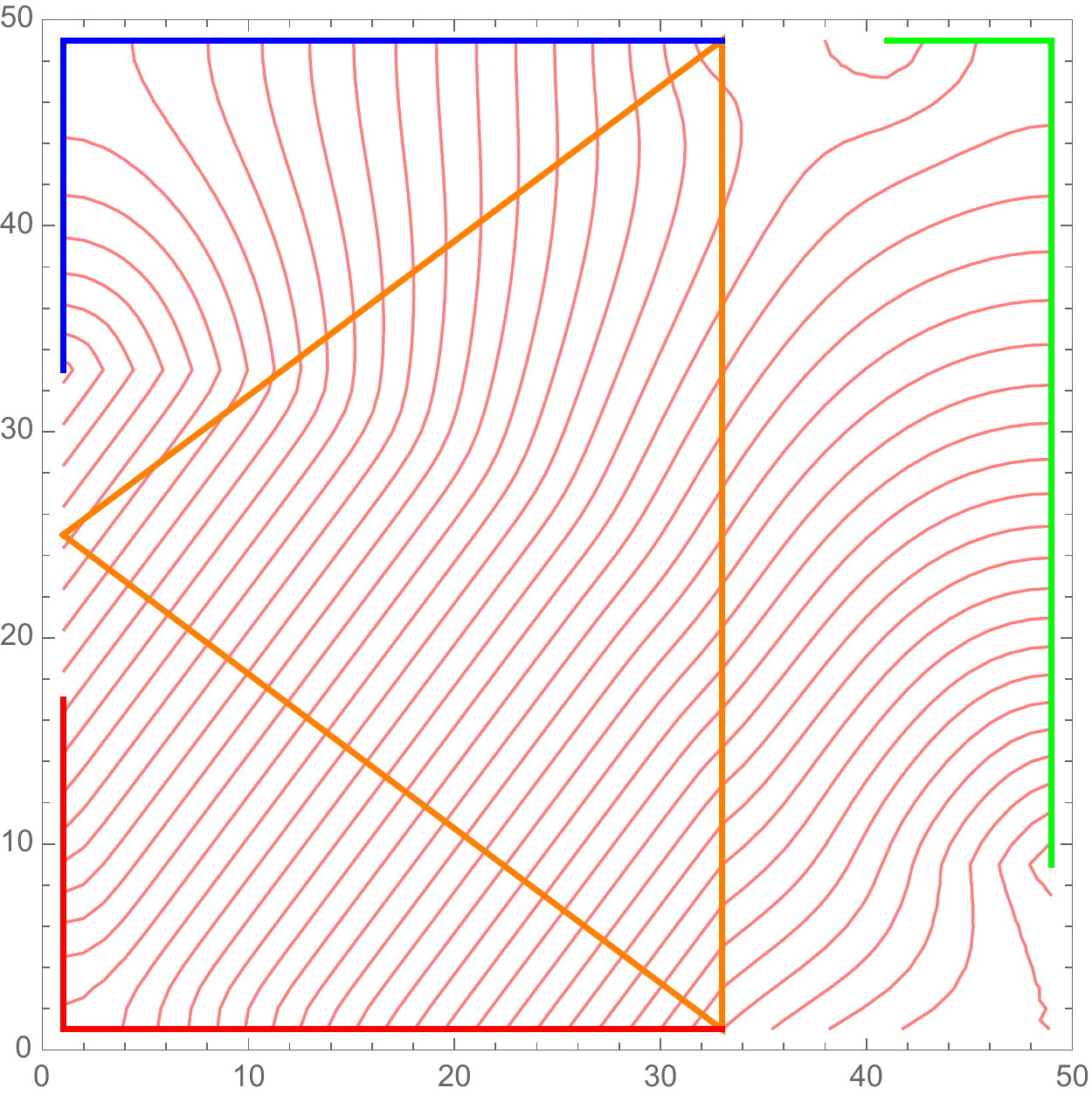} & \includegraphics[width=.4\textwidth]{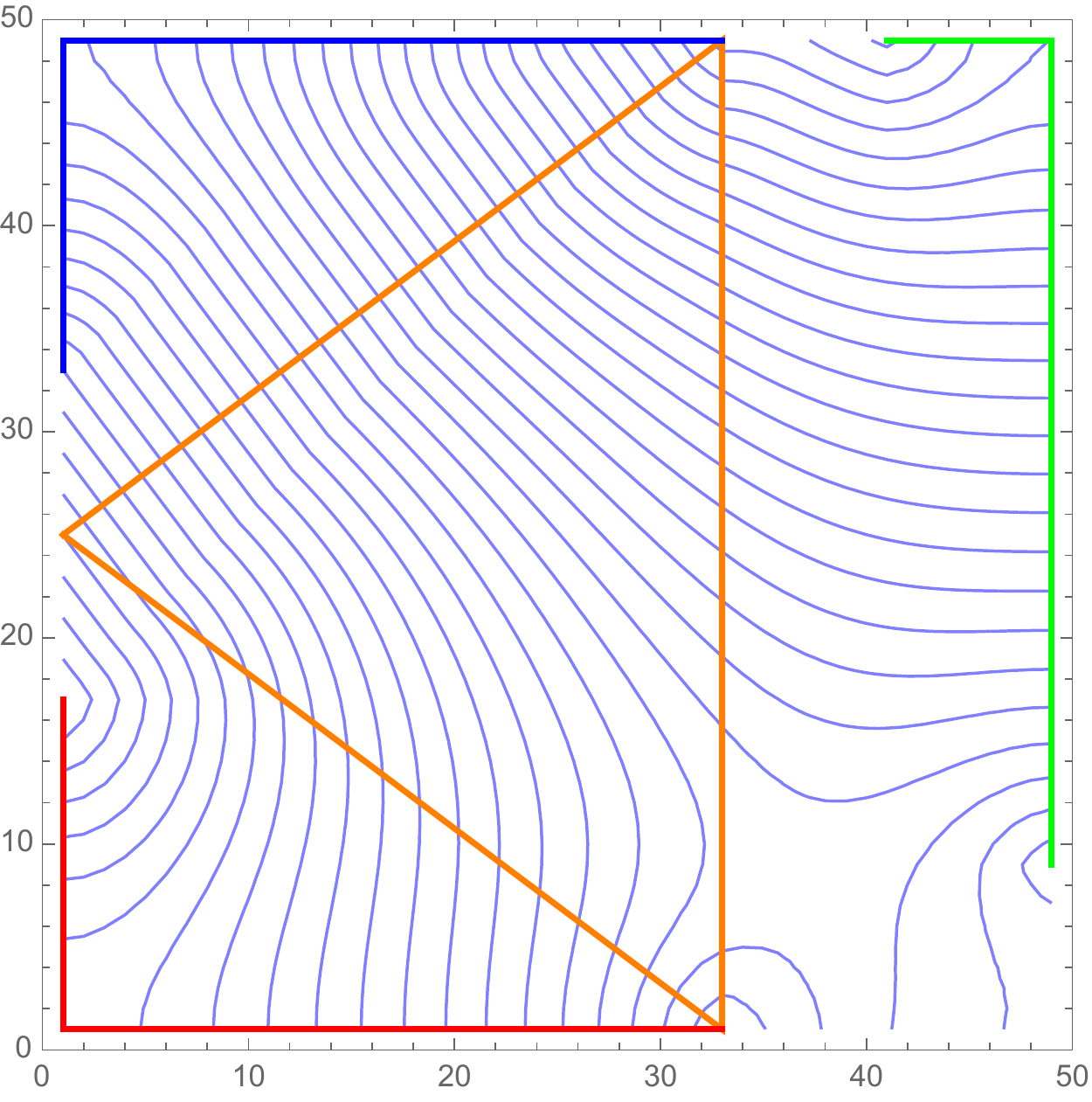} \\
$v_{A}$ & $v_{B}$ \\[6pt]
\multicolumn{2}{c}{\includegraphics[width=.4\textwidth]{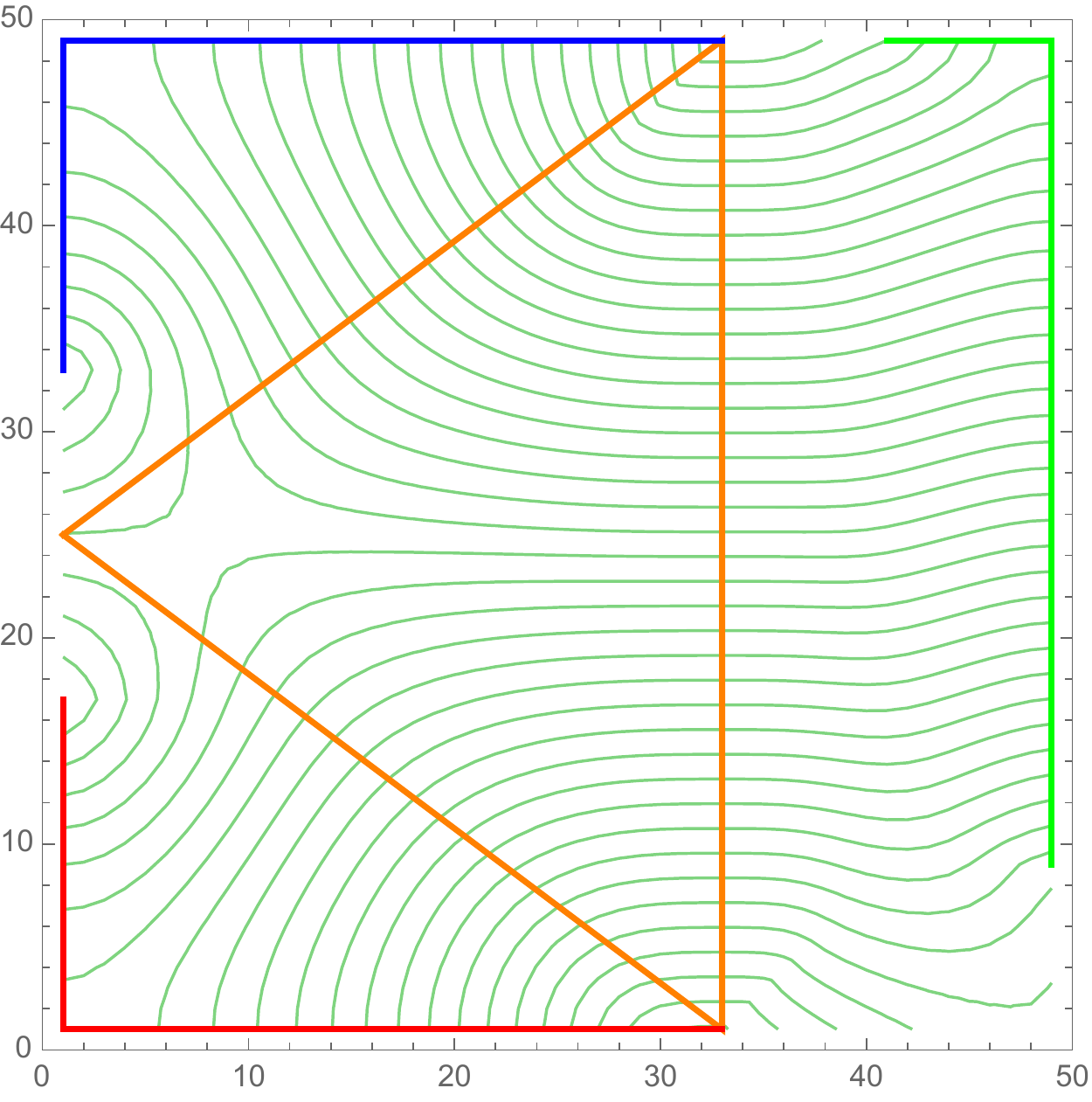} }\\
\multicolumn{2}{c}{$v_{C}$}
\end{tabular}
\caption{\label{fig:flowexample}An example of maximal flows $\mathcal{V}$ for a choice of boundary regions $A,B,C$ on a flat square. The region $A$ is in red, $B$ in blue, and $C$ in green. The minimal surface $\Sigma$ is given by the orange triangle. The flows was calculated by latticizing the manifold and using \emph{Mathematica}'s built in \texttt{FindMaximum} function.}
\end{figure}

\subsection{$\mathcal{U}$ and $w$ vector fields}

Consider a change of variables by defining new vector fields
\begin{equation}
w^{\mu} = \frac{1}{n}\sum_{i}v_{i}^{\mu}, \quad u^{\mu}_{i} = v^{\mu}_{i}-w^{\mu}.
\end{equation}
The program \eqref{meopv} becomes\footnote{The condition $\sum_{i}u^{\mu}_{i}=0$ is not necessary. Consider a set of maximal vector fields such that $\sum_{i}u^{\mu}_{i}\neq0$. Letting $\frac{1}{n}\sum_{i}u_{i}=x$ we can create a new set of fields $u'_{i}=u_{i}-x$ and $w'=w+x$ which are maximal, but satisfy $\sum_{i}u'^{\mu}_{i}=0$.}
\begin{equation}\label{meopuw}
\max_{\mathcal{U},w} \; \sum_{i}\int_{A_{i}} \sqrt{h}n_{\mu}u^{\mu}_{i} \; \text{s.t. } \nabla_{\mu}u^{\mu}_{i} = 0, \; \nabla_{\mu}w^{\mu} = 0, \; |u_{i}+w| \leq 1, \; n_{\mu}u^{\mu}_{i}|_{\mathcal{O}} =0 , \; \sum_{i}u^{\mu}_{i}=0.
\end{equation}
This redefinition comes with a trade off: at the cost of introducing a new vector field $w$ which does not contribute to the objective and modifying the norm bounds we have the benefit that $u_{i}$ cannot end on $\mathcal{O}$.
It is important to note it is not possible to generically set $w=0$. As we will see this is a distinguishing feature from the bipartite case. As a result we do not have the constraint $|u_{i}| \leq 1$, but rather $|u_{i}+w| \leq 1$, that is the constitutes of $\mathcal{U}$ are \emph{not} flows. For this reason we explicitly and pedantically refer to $(\mathcal{U},w)$ as vector fields.

In order for $E_{w}(\mathcal{A})$ to have a larger maximum than $b_{p}(\mathcal{A})$ it must be possible to transport additional flux through the manifold: this is exactly the purpose of $w$. Because of the norm bound $|u_{i}+w| \leq 1$, $w$ has the property of modifying the flux that $u_{i}$ can transport. Whenever $u_{i}$ and $w$ are anti-aligned it is possible to send twice as much flux than would normally be possible, but if they are aligned $u_{i}$ must locally be zero . As such the behavior of $w$ is to originate from $\mathcal{O}$ and travel out to $\mathcal{A}$. This way each $u_{i}$ does not encounter the bottle neck at $b_{p }(A_{i})$. In doing so, however every other vector field $u_{j_\neq i}$ can send less flux through to $A_{i}$. This naturally limits the maximum flux of $w$ out to $\mathcal{A}$.

One way to interpret this is to imagine we have chosen a configuration for $w$. Then each $u_{i}$ will experience a different effective geometry $g_{\mu\nu}^{i}$. When $u_{i}$ goes against $w$ there is twice as much space, but when $u_{i}$ tries to flow with $w$ there is none. In the effective geometry each $u_{i}$ can be renormalized to be a flow, but it is impossible to do this for the vector fields as a whole.

An interesting feature of this framework is that all of the multipartite features have been isolated to $w$, that is one vector field gives all the information needed to understand the deviation from the area of the bipartite surfaces. This can be seen in the numerical example fig.\ \ref{fig:flowexample2} which is fig.\ \ref{fig:flowexample} shown in terms of the $u$ and $w$ vector fields.

\begin{figure}[H]
\centering
\begin{tabular}{cc}
\includegraphics[width=.4\textwidth]{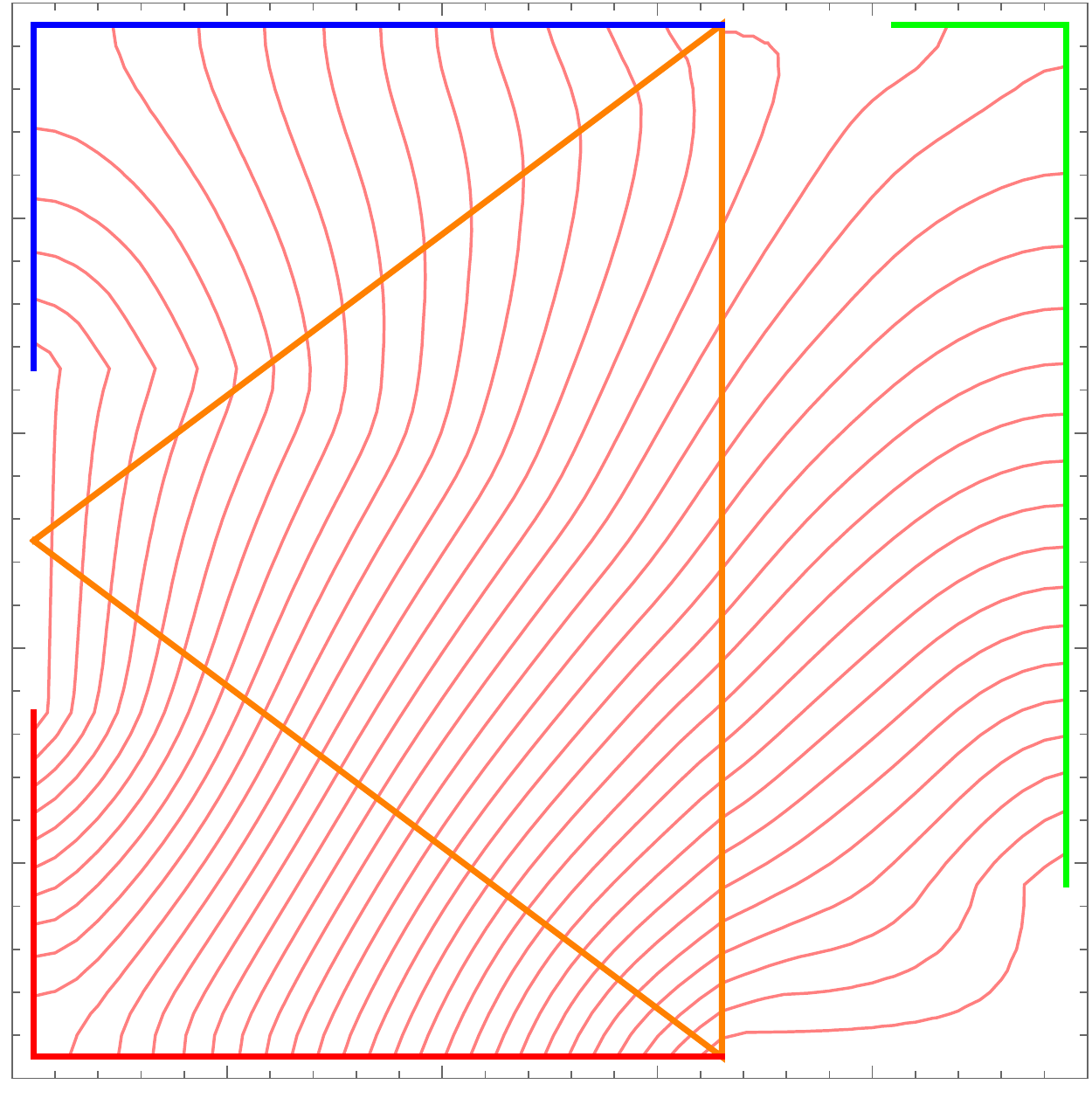} & \includegraphics[width=.4\textwidth]{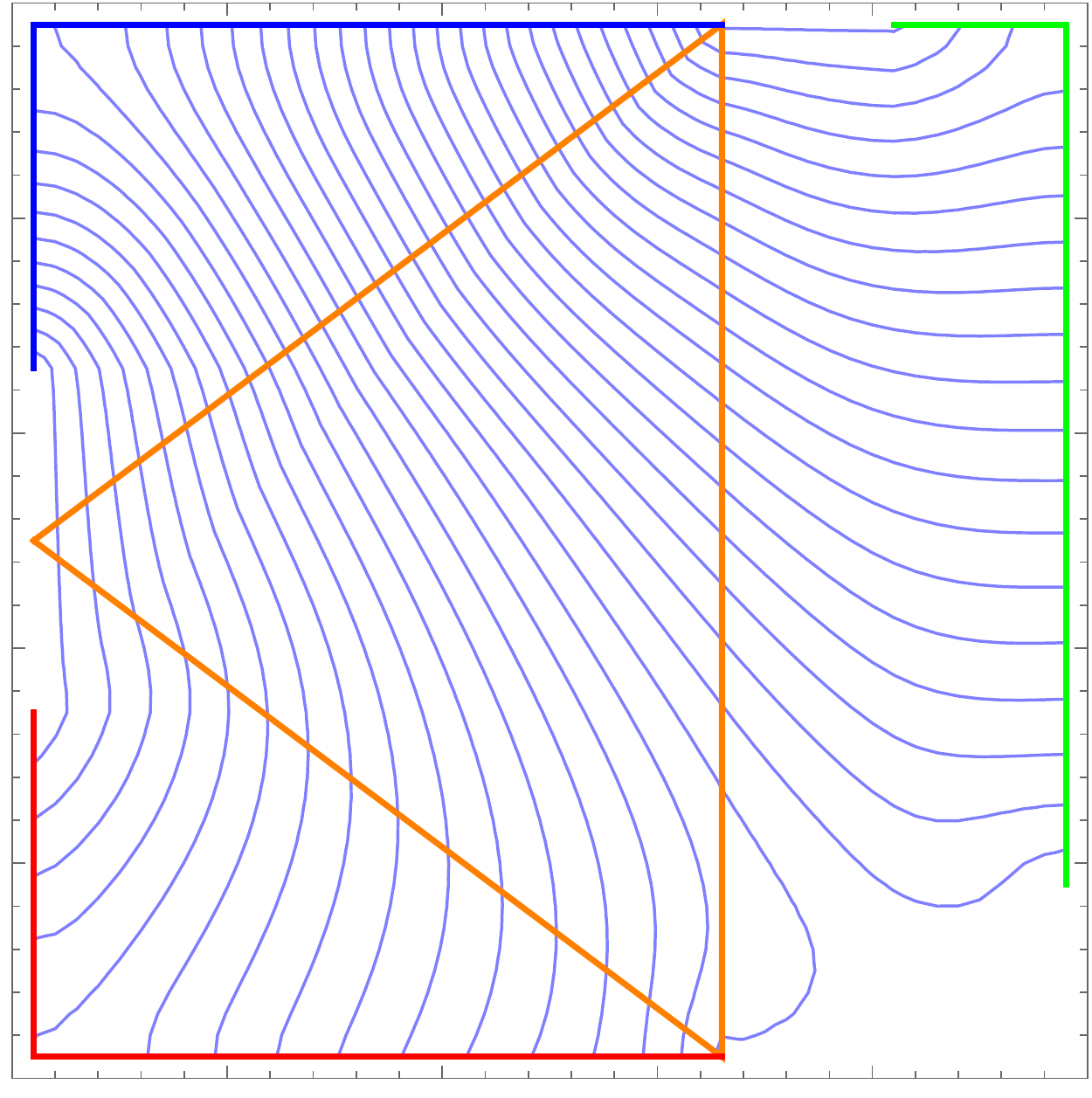} \\
$u_{A}$ & $u_{B}$ \\[6pt]
\includegraphics[width=.4\textwidth]{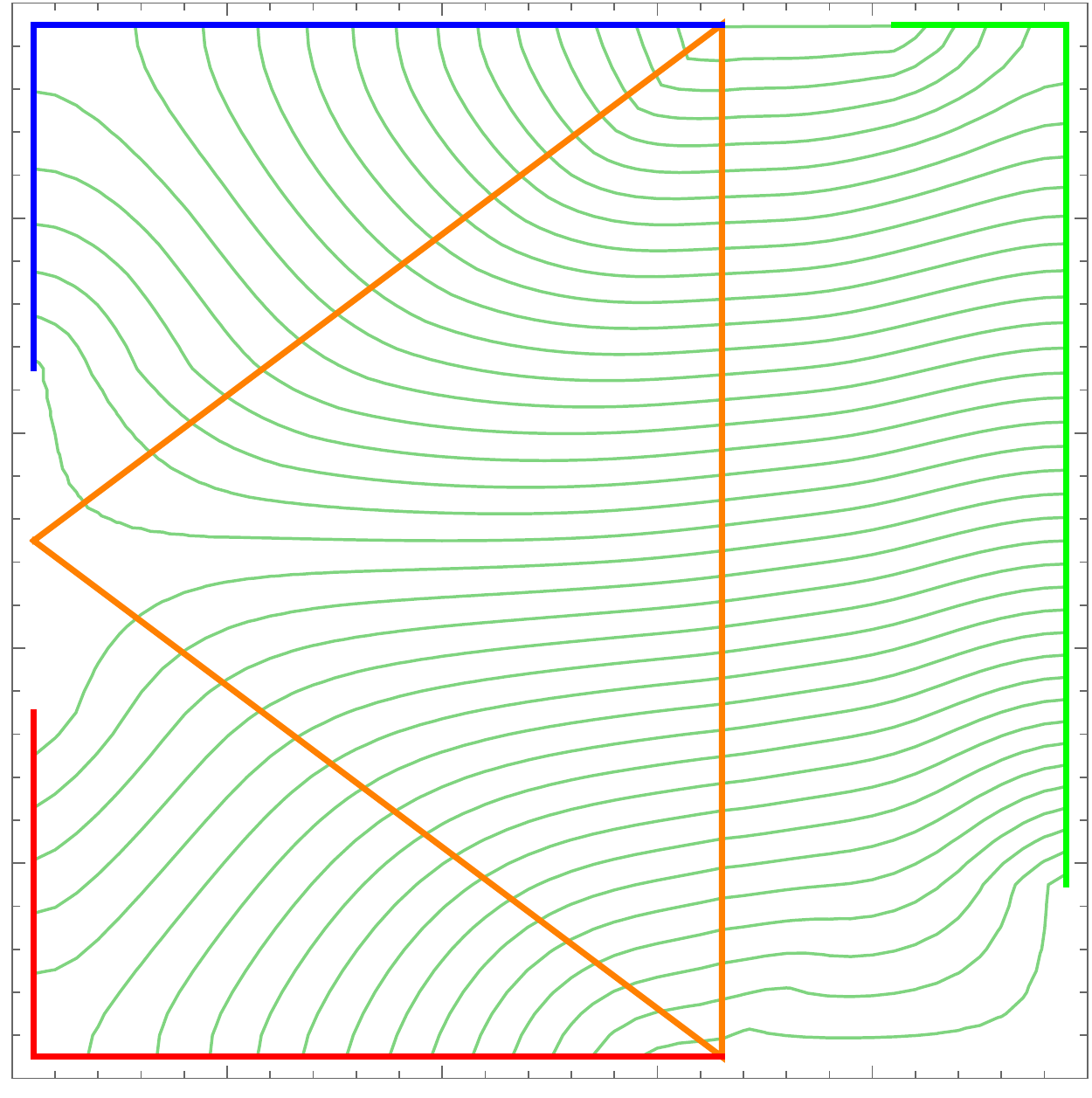} & \includegraphics[width=.4\textwidth]{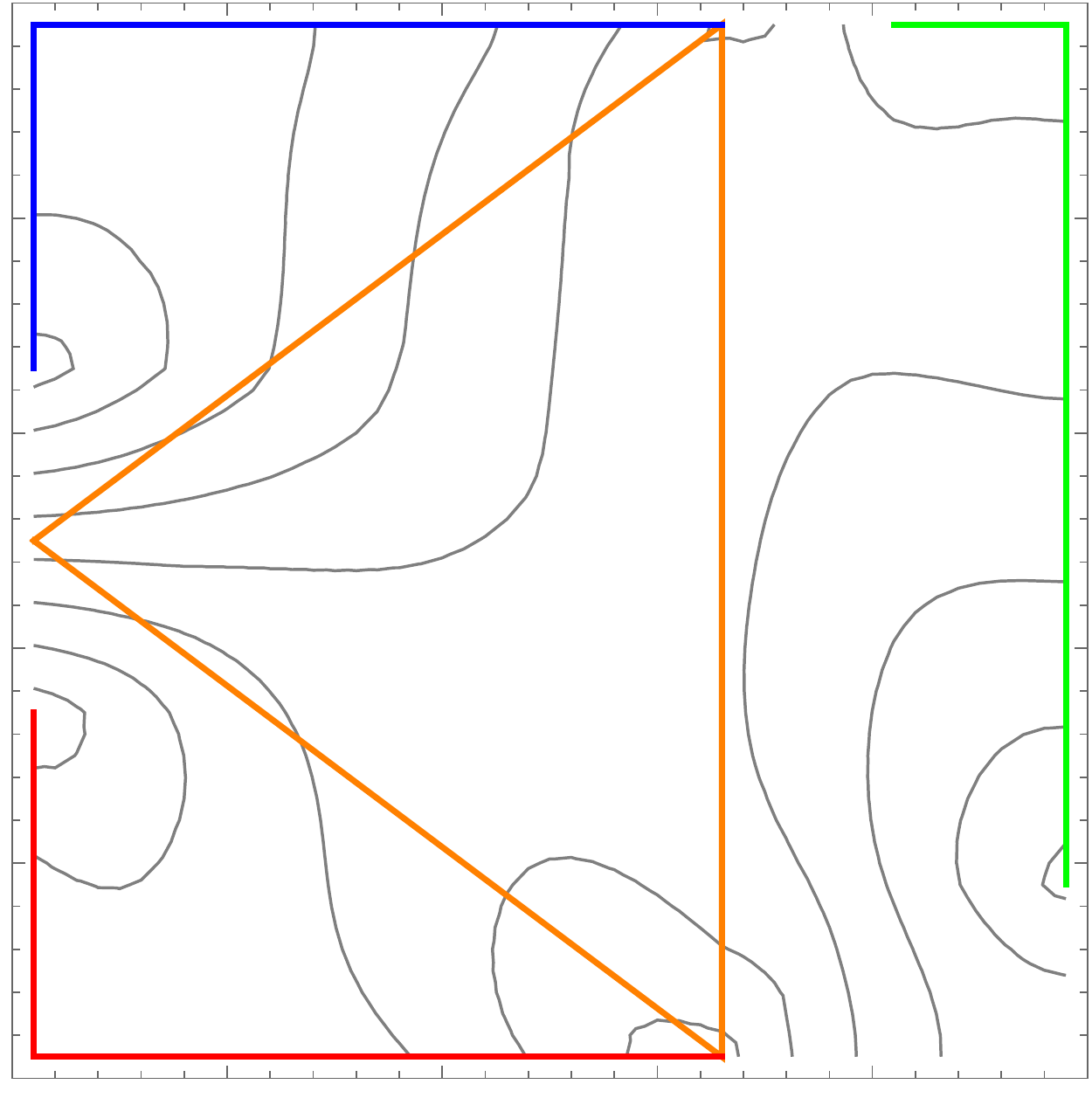} \\
$u_{C}$ & $w$ \\[6pt]
\end{tabular}
\caption{\label{fig:flowexample2}The same configuration shown in fig.\ \ref{fig:flowexample}, in terms of the vector fields $\mathcal{U}$ and $w$}
\end{figure}

\subsection{The bipartite case}\label{sec:bipartitiederv}
Having used our intuition and general properties of MFMC to define the flow formulation for $E_{w}(A:B)$ we will now provide a direct proof using convex duality which follows as a special case our multipartite derivation. Consider two regions $A$ and $B$ with RT surface $\mathcal{O}=m(AB)$ such that $\partial r(AB) = A \cup B \cup \mathcal{O}$. The minimal homology region cross section is given by
\begin{equation}
\min_{b_{p}(A:B)\sim A \, \text{rel} \, m(AB)} \area(b_{p}(A:B)).
\end{equation}

\begin{figure}[H]
\centering
\includegraphics[width=.4\textwidth]{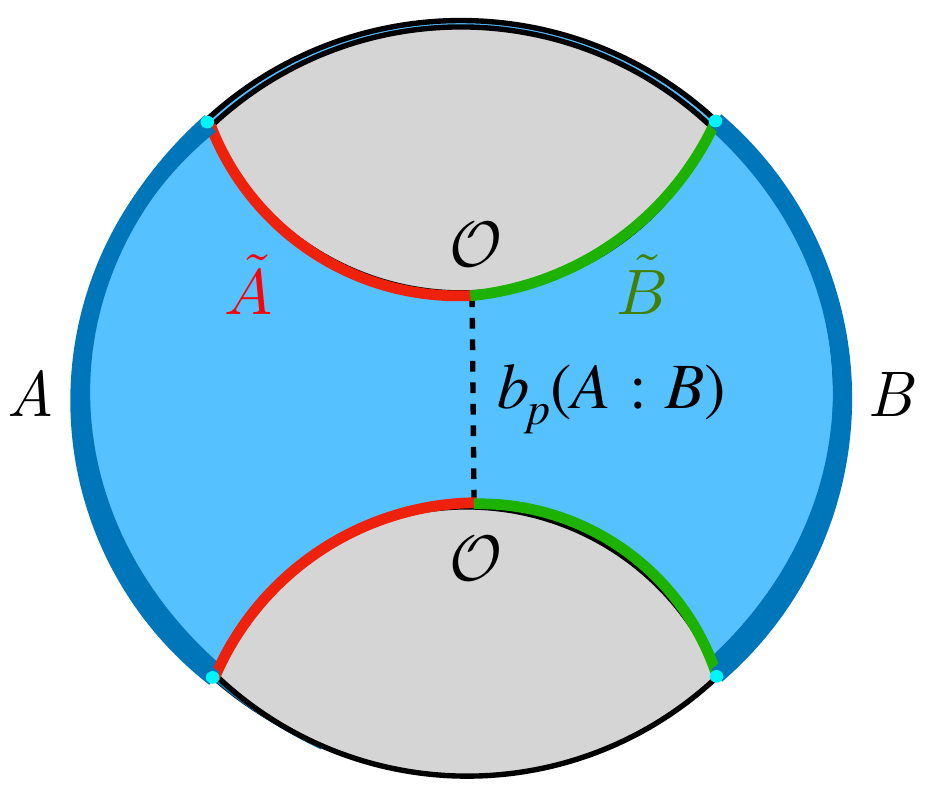}
\caption{\label{fig:partition}The bipartite homology region cross section.}
\end{figure}

Since the surface $b_{p}(A:B)$ naturally splits $\partial r(AB)$ into two regions $\tilde{A}$ and $\tilde{B}$ we can equivalently calculate the minimal surface among all partitions of $\mathcal{O}$
\begin{equation}
\min_{\tilde{A}\tilde{B}} \area(b^{\min}_{AB}) = \area(b_{p}(A:B))
\end{equation}
where $b^{\min}_{AB}$ is the minimal surface for a particular (not necessarily minimal) partition (see fig.\ \ref{fig:partition}).

Setting $n=2$ in \eqref{meopv} we can write the flow program as
\begin{equation}
\begin{split}
\max_{v^{\mu}_{A} \; v^{\mu}_{B}} \int_{A}\sqrt{h} n_{\mu}v^{\mu}_{A} + \int_{B}\sqrt{h} n_{\mu}v^{\mu}_{B} +\frac{1}{2}\int_{\mathcal{O}}\sqrt{h}n_{\mu}\left(v_{A}^{\mu}+v_{B}^{\mu}\right) \; \text{s.t. } \\ \nabla_{\mu}v_{A}^{\mu} = 0, \; \nabla_{\mu}v_{B}^{\mu} = 0, \; |v_{A}| \leq 1, \; |v_{B}| \leq 1, \; n_{\mu}v^{\mu}_{A}|_{\mathcal{O}} = n_{\mu}v^{\mu}_{B}|_{\mathcal{O}}.
\end{split}
\end{equation}

\begin{figure}[H]
\centering
\includegraphics[width=.75\textwidth]{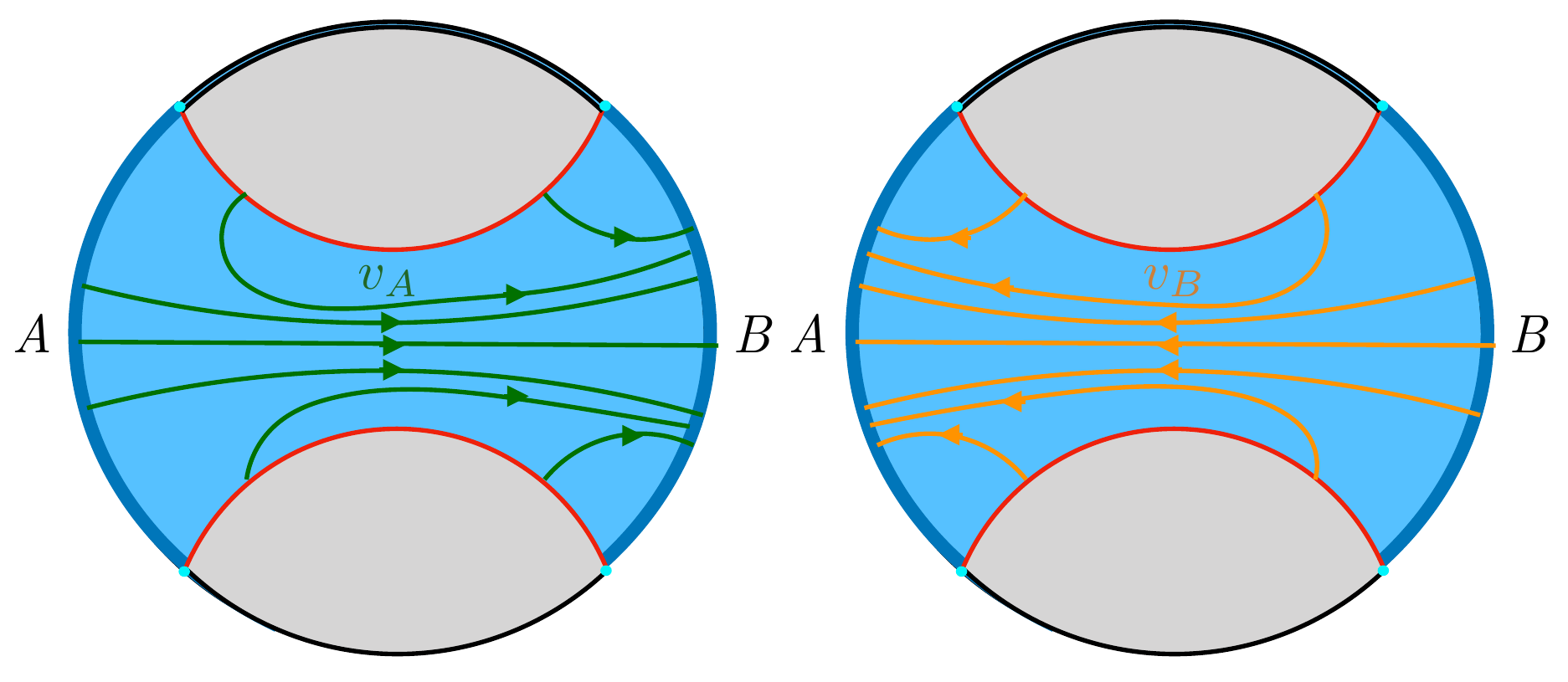}
\caption{\label{fig:vavb}An example of a general flows $v_{A}$ and $v_{B}$. Such flows are not in $\mathcal{S}'$.}
\end{figure}

The minimal homology region cross section is calculated by maximizing the flux of two flows $v^{\mu}_{A}$ and $ v^{\mu}_{B}$. These flows have the property that they are allowed to have flux on $\mathcal{O}$, but the local flux on $\mathcal{O}$ of the two flows is constrained to be the same. Such boundary threads are only counted with half the contribution of threads between $A$ and $B$ (see fig.\ \ref{fig:vavb}). Note the regions of $A$ and $B$ are large enough to source enough flow to saturate the minimal surface. As such restricting to $\mathcal{S'}$ amounts to considering those flows for which the flux on $\mathcal{O}$ is zero.

In order to make this freedom explicit we can use Stoke's theorem and the divergencelessness of the flows to find
\begin{equation}
\frac{1}{2}\int_{\mathcal{O}}\sqrt{h}n_{\mu}\left(v_{A}^{\mu}+v_{B}^{\mu}\right) = -\frac{1}{2}\left(\int_{A}\sqrt{h}n_{\mu}\left(v_{A}^{\mu}+v_{B}^{\mu}\right) + \int_{B}\sqrt{h}n_{\mu}\left(v_{A}^{\mu}+v_{B}^{\mu}\right)\right).
\end{equation}
Substituting this in the objective becomes
\begin{equation}
= \frac{1}{2}\left(\int_{A}\sqrt{h}n_{\mu}\left(v_{A}^{\mu}-v_{B}^{\mu}\right)+\int_{B}\sqrt{h}n_{\mu}\left(v_{B}^{\mu}-v_{A}^{\mu}\right)\right).
\end{equation}
To continue we make a change of variables by defining new vector fields
\begin{equation}
w^{\mu} = \frac{1}{2}\left(v_{A}^{\mu}+v_{B}^{\mu}\right), \quad u_{A}^{\mu} = v_{A}^{\mu}-w^{\mu}, \quad u_{B}^{\mu} = v_{B}^{\mu}-w^{\mu}
\end{equation}
where by definition
\begin{equation}\label{sumu}
u^{\mu}_{A}+u^{\mu}_{B} = 0.
\end{equation}

\begin{figure}[H]
\centering
\includegraphics[width=.4\textwidth]{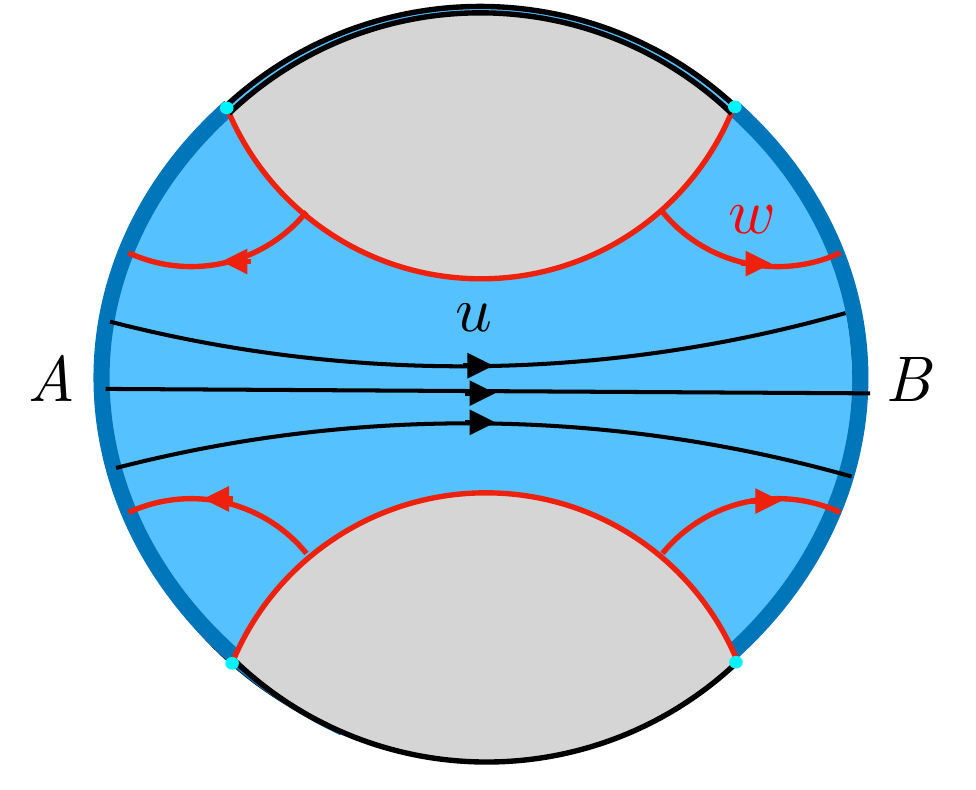}
\caption{\label{fig:uw}An example of vector fields $u$ and $w$.}
\end{figure}

The program becomes (see fig.\ \ref{fig:uw}):
\begin{equation}
\begin{split}
\max_{u_{A},u_{B},w} \;& \int_{A} \sqrt{h}n_{\mu}u_{A}^{\mu} +\int_{B} \sqrt{h}n_{\mu}u_{B}^{\mu} \; \text{s.t. } \\ &\nabla_{\mu}u_{A}^{\mu} = 0, \;\nabla_{\mu}u_{B}^{\mu} = 0, \; \nabla_{\mu}w^{\mu} = 0, \; |u_{A}+w| \leq 1, \; |u_{B}+w| \leq 1, \; n_{\mu}u^{\mu}|_{\mathcal{O}} =0.
\end{split}
\end{equation}
Because $u^{\mu}_{A}=-u^{\mu}_{B}$ a result of \eqref{sumu} the symmetry of the problem dictates $w$ can not pass through the minimum surface and therefore is unable to help increase the maximum. Thus, we can freely set $w^{\mu} = 0$ when we maximize (doing so restricts us to $\mathcal{S}'$). Accounting for normalization we have the final program
\begin{equation}
\max_{u_{A}} \; \int_{A} \sqrt{h}n_{\mu}u_{A}^{\mu} \; \text{s.t. } \nabla_{\mu}u_{A}^{\mu} = 0, \; |u_{A}| \leq 1, \; n_{\mu}u_{A}^{\mu}|_{\mathcal{O}} =0
\end{equation}
which matches \eqref{floweop}.

\subsection{Weak duality}
To understand the role of the surface $\Sigma(\mathcal{A})$ consider an arbitrary set of vector fields $(\mathcal{U},w)$ which are not necessarily maximal. By relative homology we know flux on the boundary for each $u_{i}$ must be the same as that of $\Sigma(A_{i})$. Furthermore, because $\Sigma(\mathcal{A})$ is a closed surface we must have $\int_{\Sigma(\mathcal{A})}w=0$. Thus
\begin{equation}
\sum_{i}\int_{A_{i}} u_{i} = \sum_{i}\int_{\Sigma(A_{i})}u_{i} = \sum_{i}\int_{\Sigma(A_{i})}u_{i} + \int_{\Sigma(\mathcal{A})}w = \sum_{i}\int_{\Sigma(A_{i})}u_{i} +w\,.
\end{equation}
Because of the norm bound $|u_{i}+w| \leq 1$ we have $\int_{\Sigma(A_{i})}u_{i} +w \leq \area(\Sigma(A_{i}))$. So we have established weak duality:
\begin{equation}
\sum_{i}\int_{A_{i}}u_{i} \leq \area(\Sigma(\mathcal{A})) \quad \forall \; (\mathcal{U},w)\,.
\end{equation}

\subsection{The flow from $\mathcal{O}$ is truly multipartite}
We would now like to interpret the quantity
\begin{equation}
 \int_{\mathcal{O}} \alpha = \int_{\mathcal{O}}w =  \int_{\mathcal{O}}v_{i}
\end{equation} 
for maximal flow configurations of \eqref{meopv}. By considering only solutions in $\mathcal{S}'$ we guarantee every thread of $w$ is purposefully utilized by $\mathcal{U}$ to achieve the maximum. That is there are no extraneous threads of $w$. This is obvious as a reduction in $\int_{\mathcal{O}}w$ further by deleting integral curves of $w$ would necessarily mean such a collection of vector fields was not maximal.
If $w$ is zero we know $\mathcal{U}$ are flows and $b_{p}(\mathcal{A})$ is the bottleneck. In this case $u_{i}$ saturates on $b_p(A_{i})$. When $w$ is nonzero, because the configuration $(\mathcal{U},w)\in \mathcal{S}'$, for each thread of $w$ which crosses $b_p(A_{i})$, an additional thread of $u_{i}$ can be transported across. This means the maximum flux across $b_{p}(\mathcal{A})$ is given by
\begin{equation}
\area(b_{p}(\mathcal{A})) + \int_{\mathcal{O}}w = \sum_{i}\int_{b_p(A_{i})}u_{i}.
\end{equation}
For maximal vector field configurations, $|u_{i}+w|$ saturates on $\Sigma(A_{i})$, so that
\begin{equation}
\area(\Sigma(\mathcal{A})) = \int_{\Sigma(\mathcal{A})}w + \sum_{i}\int_{\Sigma(A_{i})}u_{i} = \sum_{i}\int_{\Sigma(A_{i})}u_{i}
\end{equation}
where in the last step we have used $\int_{\Sigma(\mathcal{A})}w=0$. Now since $b_{p}(\mathcal{A})$ is homologous to $\Sigma(\mathcal{A})$ (relative to $\mathcal{O}$), we have shown
\begin{equation}
\int_{\mathcal{O}}w= \area(\Sigma(\mathcal{A}))-\area(b_{p}(\mathcal{A})).
\end{equation}
This means the flux on $\mathcal{O}$ through the manifold for a given geometry is a fixed quantity for all maximal thread configurations limited to the solution space $\mathcal{S}'$. The magnitude of the flux on $\mathcal{O}$ measures precisely the difference between the multipartite and bipartite holographic EOP proposals. In other words this flux represents how much of the holographic multipartite EOP is truly multipartite.

\subsection{Thread orientation}
So far we have described bit threads as oriented integral curves of a vector field (usually a flow). We would now like to show that it is possible to instead maximize over all valid thread configurations without reference to orientation. This can be done essentially because it is always possible to impose a standard orientation which will maximize a given thread configuration. Consider a maximal set of flows $\mathcal{V} \in \mathcal{S'}$ which calculates $E_{w}(\mathcal{A})$. Now focusing on a single flow $v_{i}$ forget about the orientation.
There are two classes of threads to consider: those between $A_{i}$ and another interval $A_{j}$ and those between $\mathcal{A}$ and $\mathcal{O}$.

For threads between $A_{i}$ and $A_{j}$ since the flux in  \eqref{meopv} is being counted on $A_{i}$ the orientation should always be imposed so that the thread exits $A_{i}$ and ends on $A_{j}$. Each such thread contributes $1$ to the objective.

Since the flux on $\mathcal{O}$ must be the same for all of $\mathcal{V}$ any thread on $\mathcal{O}$ must have the same orientation for each of the flows. Suppose the threads on $\mathcal{O}$ are oriented outward and travel to $\mathcal{A}$. For $v_{i}$ the thread will avoid going back to $A_{i}$ as this would contribute $\frac{1}{n}-1$ to the objective so instead the thread will go out to end on some $A_{j}$ this way each thread contributes $\frac{1}{n}$. The other possibility is that the threads connecting to $\mathcal{O}$ are oriented to end there. This allows a thread of $v_{i}$ to begin on $A_{i}$ and end on $\mathcal{O}$ and contribute $1-\frac{1}{n}$. The issue is that such a thread must necessarily pass through $b_{p}(A_{i})$ meaning one less thread can go between $A_{i}$ and $A_{j}$ which can never be maximal.

What we have shown is that for any maximal set of flows $\mathcal{V}$ the correct orientation is to have all threads flow out $A_{i}$ and $\mathcal{O}$ for each $v_{i}$. Therefore we can simply count the number of threads leaving these regions. Letting $N_{R}$ be the number of threads on region $R$ then the holographic multipartite EOP proposal can be written as

\begin{equation}
E_{w}(\mathcal{A}) = \max \sum_{A_{i}}N_{A_{i}}|_{v^{\mu}_{i}}+N_{\mathcal{O}}
\end{equation}
where the maximization is over valid flow configurations. Note that even though the boundary threads ``appear" for each flow they are only counted once.

\subsection{Bounds}
We now prove several bounds for the holographic multipartite EOP proposal using bit threads. For convenience and to give readers a feel for each construction we will freely switch between the $\mathcal{V}$ flows and the $\mathcal{U}$ and $w$ vector fields.

\paragraph{Properties}
Let $\mathcal{A}$ be a pure state that is $\cup_{i}A_{i} = \partial M$. In this case $\mathcal{O}=\emptyset$. Thus when maximizing \eqref{meopv} each flow $v_{i}$ is geometrically limited by $m(A_{i})$ so that the flux calculates $S(A_{i})$. Thus:
\begin{equation}
\text{For pure states: } E_{w}(\mathcal{A}) = \sum_{i}S(A_{i}).
\end{equation}

We will now show that $E_{w}(\mathcal{A})$ can only be smaller after discarding part of one of the subregions. This can be thought of intuituvely as reducing the amount of space the threads can occupy. Let $\mathcal{A'} = \{X\cup A_{1},A_{2},...\}$ we then have the two programs
\begin{equation}
\begin{split}
E_{w}(\mathcal{A'}) =& \max_{\mathcal{V},\alpha} \sum_{i}\int_{A_{i}}\sqrt{h} n_{\mu}v^{\mu}_{i} + \int_{\mathcal{O'}}\sqrt{h}\alpha \; s.t.  \; n_{\mu}v^{\mu}_{i}|_{\mathcal{O'}} = \alpha \; \text{on $r(\mathcal{A'})$}\\
E_{w}(\mathcal{A}) =& \max_{\mathcal{V},\alpha} \sum_{i}\int_{A_{i}}\sqrt{h} n_{\mu}v^{\mu}_{i} + \int_{\mathcal{O}}\sqrt{h}\alpha \; s.t.  \; n_{\mu}v^{\mu}_{i}|_{\mathcal{O}} = \alpha \; \text{on $r(\mathcal{A})$}.
\end{split}
\end{equation}
Our goal will be to write the second of these as a program on $r(A')$ with additional constraints. This can be done by viewing the program as being on $r(\mathcal{A'})$ and imposing the additional constraint that any threads must remain in or end on the boundary of $r(\mathcal{A})$. This is always possible since $r(\mathcal{A}) \subset r(\mathcal{A}')$. Therefore using Theorem \ref{T1}

\begin{equation}
E_{w}(\mathcal{A})\leq E_{w}(\mathcal{A'}).
\end{equation}

\paragraph{Upper bound}
Consider the maximal flow $v_{A_{i}}$ which saturates on $\Sigma(A_{i})$ with flux $\area(\Sigma(A_{i}))$. In order for this flux to count maximally to the objective, the flow will want to end on $\mathcal{A}\backslash A_{i}$ which is limited by the geometric obstruction $\sum_{j \neq i}b_{p}(A_{j})$. The flux is further limited by the boundary threads from the other flows in the amount $\sum_{j \neq i}\area(\Sigma(A_{j}))-\sum_{j \neq i}\area(b_{p}(A_{j}))$. From these considerations we have the bound
\begin{equation}
E_{w}(\mathcal{A}) \leq 2\min_{i}\sum_{j\neq i}\area(b_{p}(A_{j})).
\end{equation}

\begin{figure}[H]
\centering
\includegraphics[width=.4\textwidth]{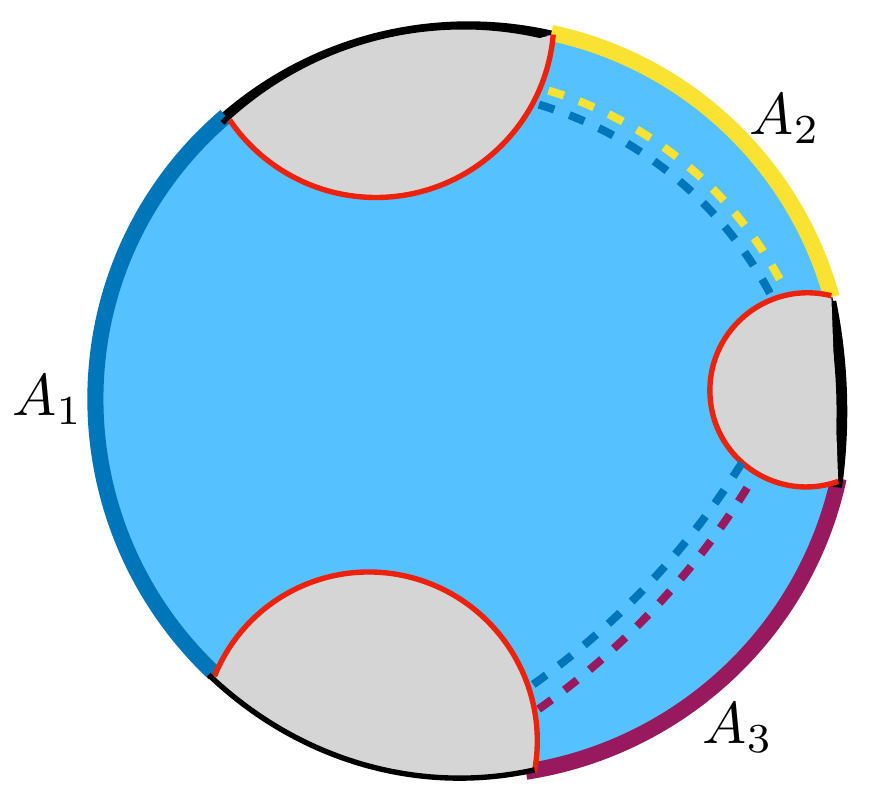}
\caption{\label{fig:upperboundmeop}In this case $\area(\Sigma(A_{1})) = \area(\Sigma(A_{2})) + \area(\Sigma(A_{3}))$ and the flux on $\mathcal{O}$ will be zero.}
\end{figure}

Note such a surface is in the same relative homology class as the surfaces we have been considering and is thus an allowed candidate when minimizing (see fig.\ \ref{fig:upperboundmeop}). This inequality saturates when one boundary region has enough capacity to send flow and saturate the bipartite surfaces of the other regions by itself. In such cases there is no flux on $\mathcal{O}$ (put another way $w=0$).

For a pure state this bound reduces to
\begin{equation}
E_{w}(\mathcal{A}) \leq \min_{i}( S(A_{i}) + \ldots + S(A_{1}\ldots A_{i-1},A_{i+1}\ldots A_{n})+\ldots +S(A_{n}))
\end{equation}
which matches the bound given by \cite{Umemoto2018}. Note that part of $S(A_{1}\ldots A_{i-1},A_{i+1}\ldots A_{n})$ is part of $\mathcal{O}$ in the non pure case. In our construction this surface can never be counted indicating this bound is only saturated in the case of purity.

\paragraph{Lower bounds}

\begin{figure}[H]
\centering
\includegraphics[width=.4\textwidth]{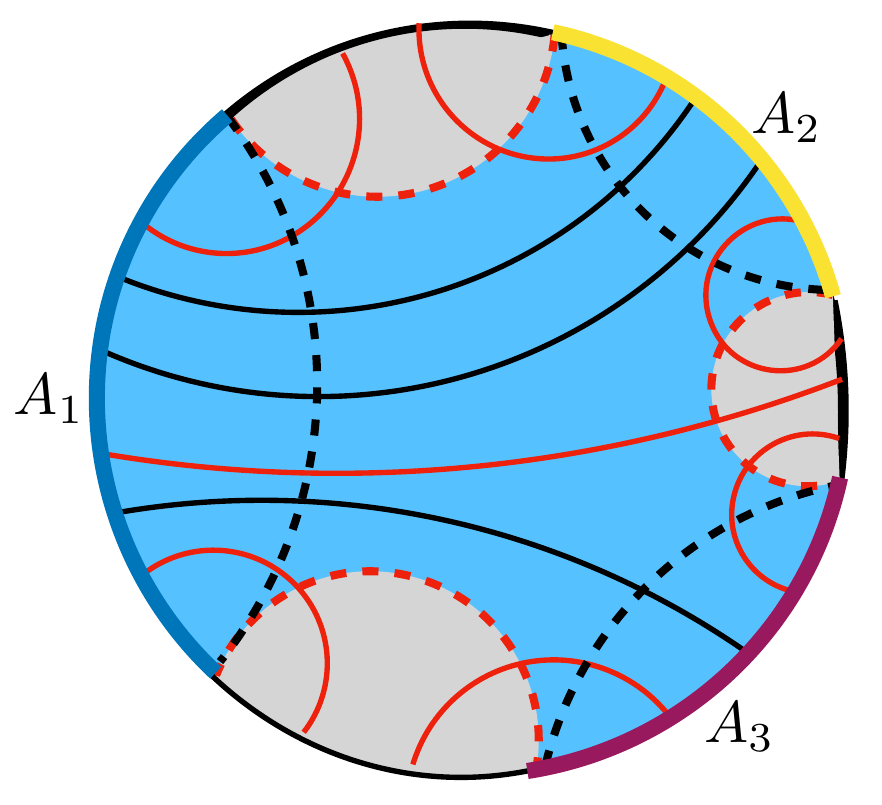}
\caption{\label{fig:lb} Minimal surfaces are shown as dashed lines and flows as solid lines. The maximum multiflow the manifold can support is represented here, but not all flow lines are shown as each dashed line is saturated. To create a valid flow configuration in the feasible set of \eqref{meopv} we delete each thread shown in red.}
\end{figure}

As in the bipartite case, we will derive a lower bound on the holographic multipartite EOP proposal via an explicit construction. As previously described the maximum total flux a multiflow can support is the sum of the entropies of the regions whose union is $\partial M$. To get a valid flow configuration of $E_{w}(\mathcal{A})$ we can choose to look at multiflows which remain in $r(\mathcal{A})$. The simplest way (though not optimal) to create such a flow is to delete each thread which passes through $m(\mathcal{A})$. The number of times such threads leave $r(\mathcal{A})$ is exactly $S(A_{1}\ldots A_{n})$ which on the boundary are each counted twice (see fig.\ \ref{fig:lb}). That is, we have shown the holographic multipartite EOP proposal is bounded below by the multipartite information \footnote{The explicit construction of this flow provides evidence that an information theoretic quantity the squashed entanglement will always be saturated holographically. See \ref{sec:squashed} for more details.}
\begin{equation}
E_{w}(\mathcal{A}) \geq I(\mathcal{A}) \coloneqq \sum_{i} S(A_{i}) - S(A_{1}\ldots A_{n}).
\end{equation}

To get another lower bound consider the set of flow configurations of \eqref{meopuw}. If we restrict ourselves to only flow configurations with $w=0$ then the geometric barrier $b_{p}(\mathcal{A})$ is unavoidable. Thus, if one explicitly imposes the condition $w=0$ to \eqref{meopuw} the resulting program will calculate $b_{p}(\mathcal{A})$. So by Theorem \ref{T1} we have the bound
\begin{equation}
E_{w}(\mathcal{A}) \geq \area(b_{p}(\mathcal{A}))
\end{equation}
which was first shown in \cite{Umemoto2018}.

\section{Future directions}\label{sec4}

\paragraph{Flux from boundary and bulk intersections}
In this article we have utilized convex duality to derive flow descriptions of the holographic bipartite and multipartite EOP proposals. In the latter case we find the addition of a coupled boundary flow which after restricting to $\mathcal{S}'$ is truly multipartite. Upon examination of the dualization process this is not a unique feature, but rather should be present in the flow dual to any surface with nontrivial boundary intersection. Similarly, as can be seen in \ref{sec:bulk}, bulk intersection points will generically lead to a bulk flux as indicated by the vector fields having a coupled non-zero divergence. It would be interesting to explore the generality of these features in more detail.

\paragraph{Dynamics}
The maximin construction of \cite{Wall_2014} allows for covariant HRT surfaces \cite{Hubeny_2007} to be calculated by minimizing a surface on a fixed time slice and then maximizing over all time slices. Since our construction of holographic EOP proposal holds for a fixed time slice it follows one could define analogously the covariant EOP by calculating a maximal flow and then maximizing over all possible time slices. This seems to indicate that many of the features of the holographic EOP proposal should easily generalize to the covariant case. It would be worthwhile to work out these features in detail and describe them in the context of the full covariant construction of bit threads \cite{cbt}.

\paragraph{Bit threads and entanglement distillation}
For the standard entanglement entropy a bit thread is normally described as representing a distilled EPR pair between the boundary regions it connects. This holds too for the multiflow where the collection of flows can be simultaneously placed on the manifold. The question then arises how one should interpret the flows of \eqref{meopv} for the holographic multipartite EOP proposal which should presumably be related to a distillation of the multipartite entanglement. While the bipartite contributions may retain their interpretation as EPR pairs it is less clear how to represent the threads of $\mathcal{V}$ which connect to $\mathcal{O}$ (or equivalently the ancillary vector field $w$). This is because even though each $v_{i}$ will have the same number of such threads the region on which they end can and must change amongst the various elements of $\mathcal{V}$. The need for multiple threads to describe a multipartite state is not dissimilar to the original description of the GHZ state given in \cite{Freedman2017}. Possibly one could interpret each $v_{i}$ as a distinct purification such that a collection of purifications is needed to describe $E_{p}(\mathcal{A})$. It would be very interesting to see if direct CFT calculations of $E_{p}(\mathcal{A})$ would give rise to similar features.

\paragraph{Holographically multipartite entanglement is inefficient}
We have established the bounds
\begin{equation}
\area(b_{p}(\mathcal{A})) \leq E_{w}(\mathcal{A}) \leq 2\min_{i}\sum_{j\neq i}\area(b_{p}(A_{j}))
\end{equation}
for the holographic multipartite EOP. Note that the flow configurations in both extremes have $\alpha=0$. That is, in the language of our framework, these configurations have no multipartite entanglement. This is not so surprising as the threads on $\mathcal{O}$ are rather inefficient. This is because one thread of each ``color" is required to make a ``unit" contribution. This is in contrast to the bipartite threads connecting $A_{i}$ and $A_{j}$ which only require one. The multipartite threads simply take up more space in the geometry. It would be interesting to understand why it is ``hard" to get multipartite entanglement. Potentially an understanding of the relation between the choice of boundary configurations and the maximum multipartite contribution could shed light on this feature.

\paragraph{Towards proving the conjecture}
It has been firmly established that the homology region cross section and its multipartite generalization are interesting bulk surfaces with possible information theoretic content. What remains to be done is to firmly establish the role these surfaces play in our understanding of AdS/CFT. Throughout this paper we have worked under the assumption $E_{p}(\mathcal{A})=E_{w}(\mathcal{A})$ following evidence such as the original proposals \cite{Umemoto_2018,Nguyen2018,Umemoto2018,Bao:2018aa} and work towards a proof \cite{Bao:2018ac,Guo:2019aa,bao:aa}. Still there is reason to proceed with some healthy skepticism. Alternate proposals exist on both sides: \cite{Dutta:aa} calculates the bipartite homology region cross section using a canonical purification whose geometry is that of a wormhole. In this formalism explicit CFT calculations were able to be done using the replica trick. Other duals such as the logarithmic negativity \cite{Kudler-Flam:2018aa} and the odd entanglement entropy  \cite{Tamaoka_2019} have appeared as well. Similarly the ``bipartite dominance" conjecture of \cite{Cui:2018aa} predicts that holographically $E_{p}(A:B)=\frac{1}{2}I(A:B)$ always. 

Almost all of the work done up to this point has focused on the bipartite case where things are much simpler. However from the bit thread perspective it is only when one is working the multipartite case that interesting features appear specifically the coupling of flows due to the boundary intersection of the minimal surface. It is our hope that this work and the characterization $\Sigma$ in terms of bit threads may be elucidating to others and provide inertia towards a complete holographic characterization that answers these many questions.

\acknowledgments
The work of J.H. is supported in part by the National Science Foundation under the IGERT: Geometry and Dynamics Award No.\ 1068620 and in part by the Simons Foundation through \emph{It from Qubit: Simons Collaboration on Quantum Fields, Gravity, and Information}. The work of M.H. is supported in part by the Simons Foundation through \emph{It from Qubit: Simons Collaboration on Quantum Fields, Gravity, and Information} and in part by the Department of Energy Office of High-Energy Physics through Award DE-SC0009987. We would like to thank Cesar Ag\'on, Ning Bao, Bartlomiej Czech, Jesse Held, Veronika Hubeny, Charles Stine, Koji Umemoto, and Mark Van Raamsdonk for useful discussion. We would also like to thank Harsha Hampapura and Jesse Held for useful comments on an earlier draft. J.H. would like to thank the Instituto Balseiro, the International Centre for Theoretical Sciences (ICTS), the Institute for Advanced Study (IAS), the Israel Institute for Advanced Study (IIAS), and the Yukawa Institute for Theoretical Physics (YITP) for hospitality during various stages of this work.

\appendix

\section{Multiflows}\label{sec:mf}
In this appendix we define the multiflow and then for convenience state its most important properties including a flow proof of monogamy of mutual information. The entirety of this construction is due to the authors of \cite{Cui:2018aa} where additional details can be found.
\begin{definition}[Multiflow]\label{multiflow}
On a Riemannian manifold $M$ with boundary $\partial M$, split $\partial M$ into n non overlapping regions $\mathcal{A} = \{A_{i}\}$ . A multiflow is a collection of vector fields $\mathcal{V}_{m} = \{v_{ij}^{\mu}\}$ which satisfy the following properties:
\begin{outline}
\1$n_{\mu}v^{\mu}_{ij} =0 \; \text{on } A_{k}, \; k \neq i,j$
\1 $v_{ij} = -v_{ij}$
\1 $\nabla_{\mu}v_{ij}^{\mu} = 0$
\1$\sum\limits_{i<j}\left|v_{ij}\right| \leq 1.$
\end{outline}
\end{definition}
We will use $\mathcal{V}$ to indicate a set of flows, but we will reserve the notation $\mathcal{V}_{m}$ to indicate a multiflow which satisfies all the conditions of \ref{multiflow}.
\begin{definition}[Subflow]
Given a multiflow $\mathcal{V}_{m}$ consider the union of any number of boundary regions in $\mathcal{A}$ call it $B$. The subflow of $B$ is the sum of all flows leaving $B$
\begin{equation}
v_{B} = \sum\limits_{i \text{ s.t } A_{i} \in B, \; j} v_{ij}
\end{equation}
\end{definition}
\begin{theorem}[Existence of a maximal multiflow]\label{existmultiflow}
There exists a multiflow $\mathcal{V}_{m}$ such that all single interval subflows and any one additional subflow, call the associated boundary region $B$, are maximal. That is
\begin{equation}
S(A_{i}) = \int_{A_{i}} v_{A_{i}} \quad S(B) = \int_{B} v_{B}
\end{equation}
for some $\mathcal{V}_{m}$. In general all other subflows will not be maximal.
\end{theorem}
\begin{theorem}[Monogamy of mutual information (MMI)] Consider the case of $4$ boundary regions $A,B,C,D$
\begin{equation} \label{MMI}
-I_{3}(A:B:C) = S(AB)+S(AC)+S(BC) -S(A)-S(B)-S(C)-S(ABC) \geq 0
\end{equation}
\end{theorem}
\begin{proof}
The proof follows almost immediately by writing out the subflows and using existence of a maximal multiflow. Consider the two interval subflows $v_{AB},v_{AC},v_{BC}$. In general these can not all be made maximal thus
\begin{equation}
S(AB)+S(BC)+S(AC) \geq \int_{AB}v_{AB} +\int_{AC}v_{AC} + \int_{BC}v_{BC}.
\end{equation}
Expanding out the subflows and using maximality of the single interval subflows one has
\begin{equation}
= \int_{A} v_{A} + \int_{B} v_{B} + \int_{C} v_{C} -\int_{ABC} v_{D} = S(A) +S(B) + S(C) + S(ABC)
\end{equation}
Subtracting the two sides gives MMI.
\end{proof}

\section{The integer relaxation}\label{sec:relax}

\begin{figure}[H]
\centering
\includegraphics[width=.5\textwidth]{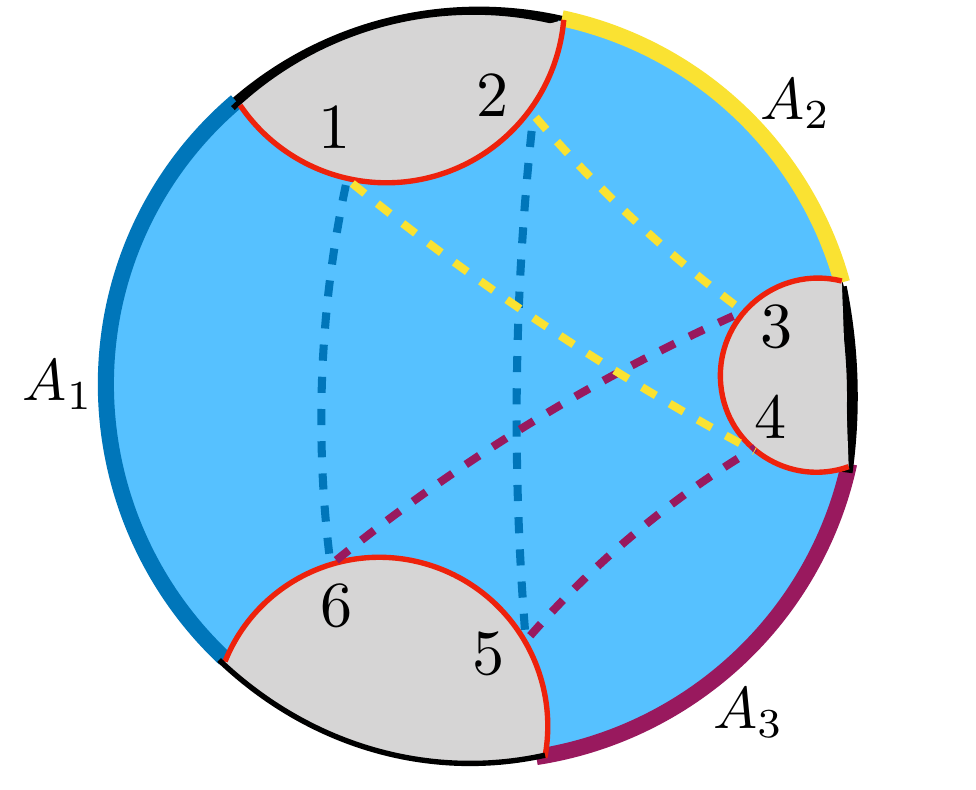}
\caption{\label{fig:doubletriangle}An allowed surface configuration after the integer relaxation.}
\end{figure}

\begin{figure}[H]
\begin{tabular}{ccc}
\includegraphics[width=.333\textwidth]{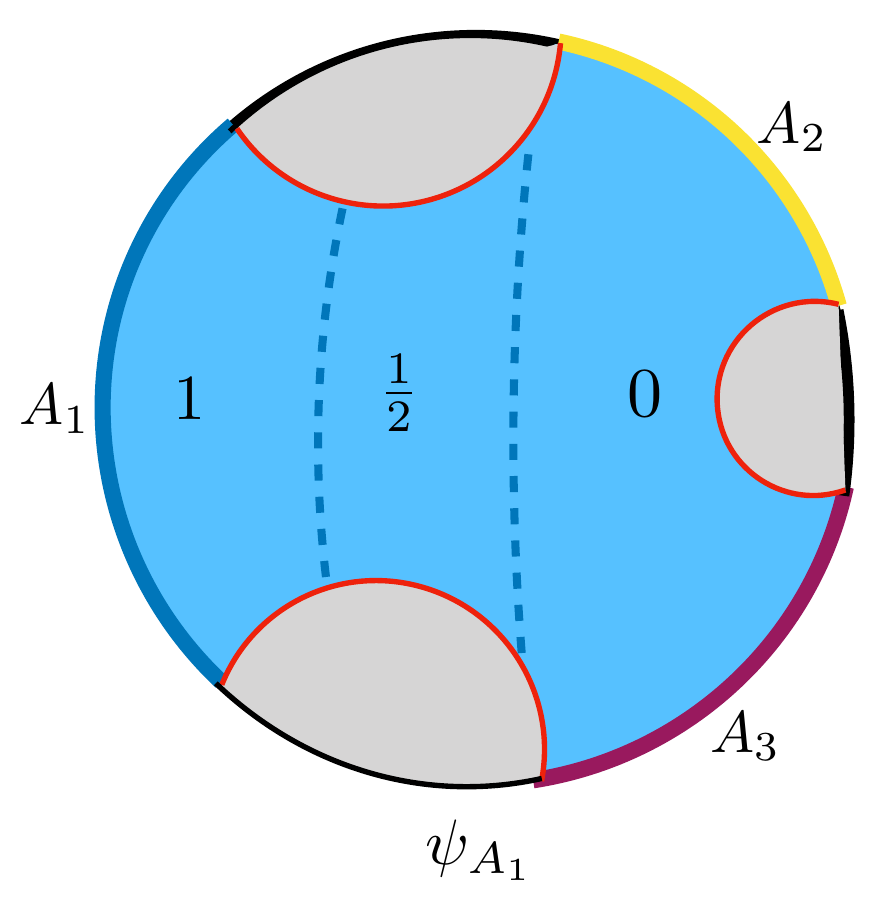} & \includegraphics[width=.333\textwidth]{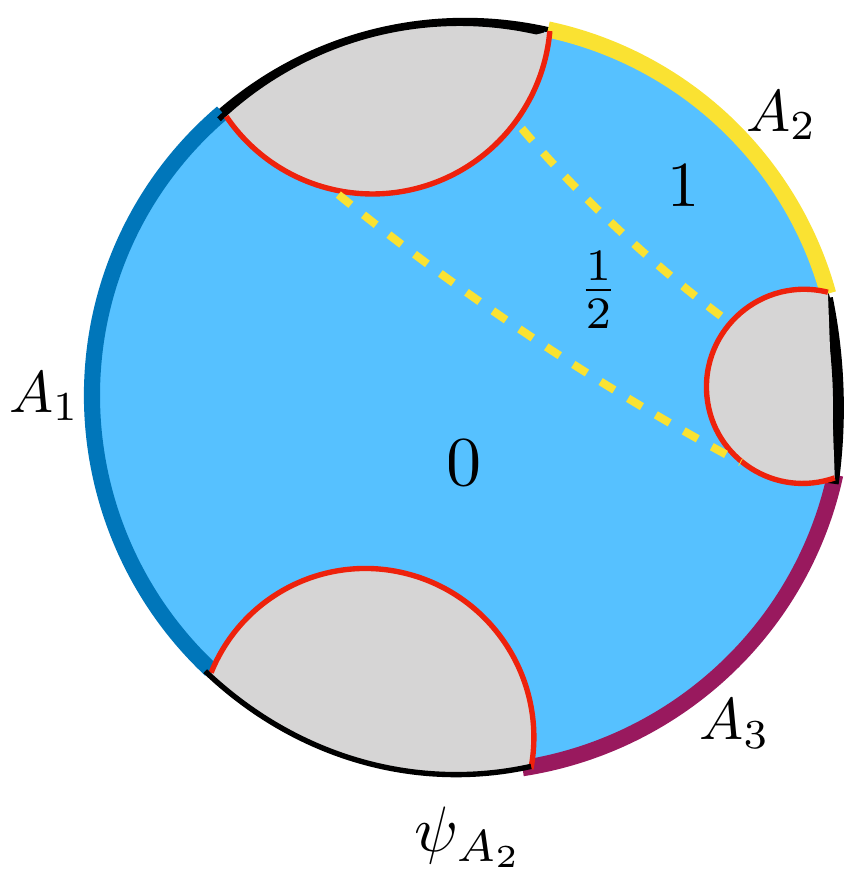} & \includegraphics[width=.333\textwidth]{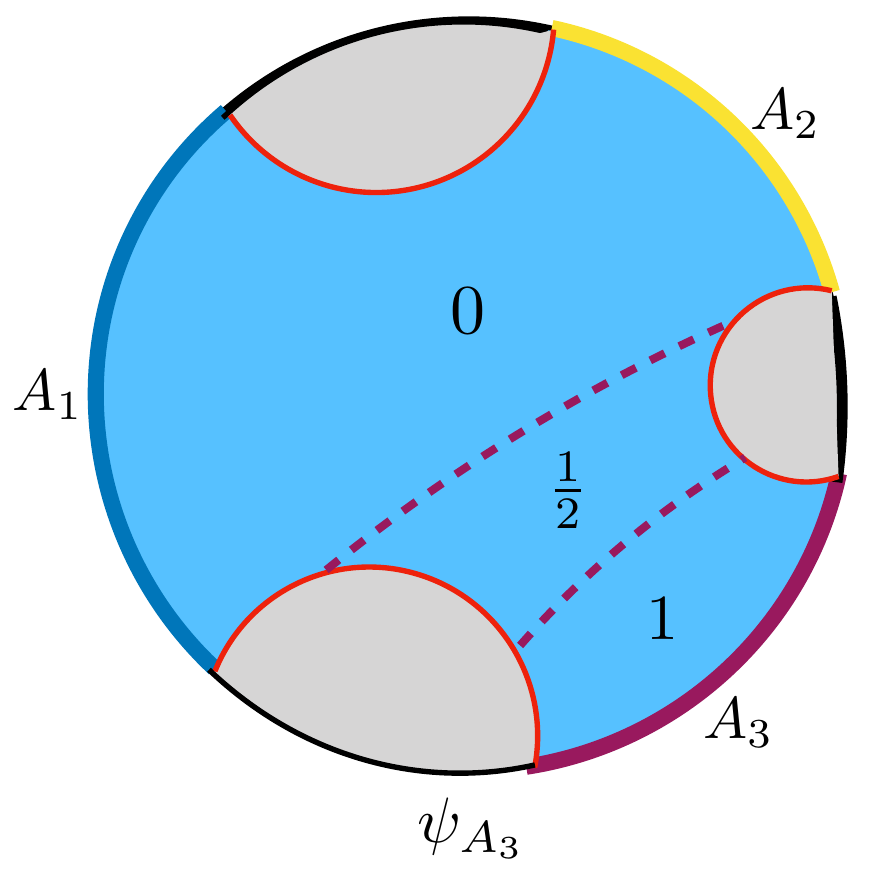}\\
\end{tabular}
\caption{\label{fig:psiconfig}The corresponding configuration of $\psi_{i}$'s.}
\end{figure}

In our derivation of the the flow program for the holographic multipartite EOP we were required to make the the following relaxation of the primal program
\begin{equation}
\text{On $\mathcal{O}$: } \psi_{i}  \in \{0,1\} \Rightarrow \psi_{i}  \in \mathbb{R}.
\end{equation}
Notably this allows solutions where each $\psi_{i}$ changes in more than one location and certain portions of $\mathcal{O}$ are shared by multiple $\psi_{i}$'s (for example consider the configuration of surfaces and the corresponding $\psi_{i}$'s given by figs.\ \ref{fig:doubletriangle} and \ref{fig:psiconfig}).

\begin{figure}[H]
\centering
\includegraphics[width=.5\textwidth]{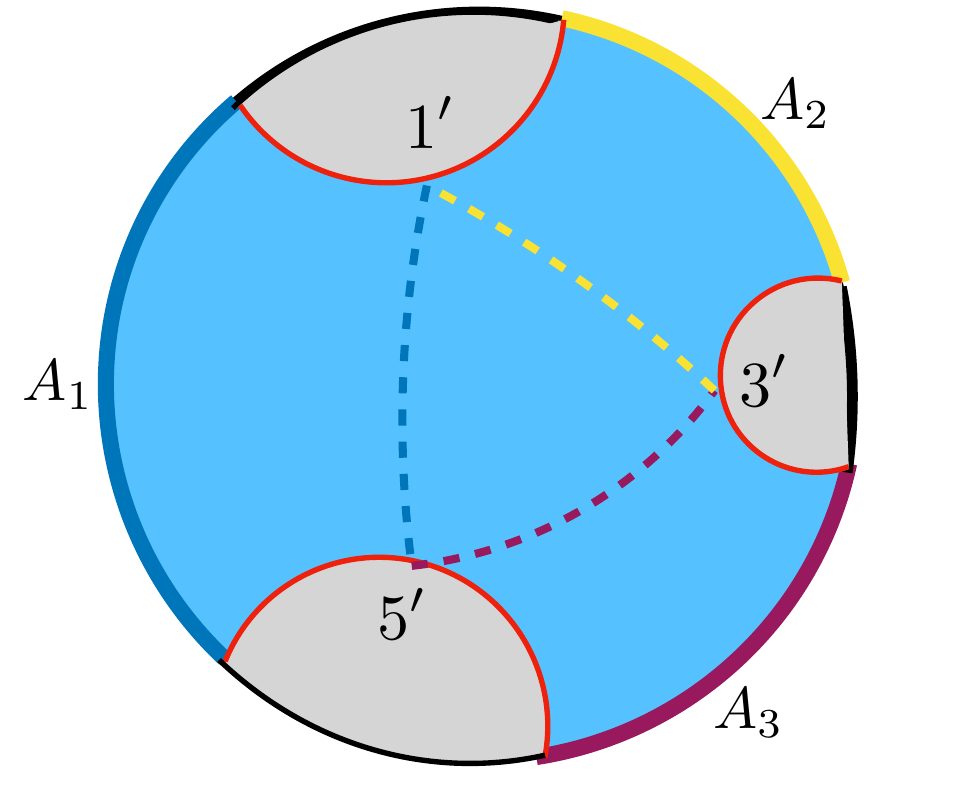}
\caption{\label{fig:singletriangle}An integer surface configuration.}
\end{figure}

It is possible to construct a geometry (though granted a highly contrived one) which ``forces" the lengths of these surfaces to be small making it the correct minimum to the program \eqref{mprelaxedagain} which would be different from the optimal configuration of \eqref{umemotoprop} (which is shown in fig.\ \ref{fig:singletriangle}) . This demonstrates that our relaxed program is no longer generically the same as the original proposal.

We would now like to prove that for a $2+1d$ bulk with everywhere negative curvature and three boundary regions the relaxed program  \eqref{mprelaxedagain} will always be the same as the original proposal \eqref{umemotoprop}. That is we can not do any better by allowing the $\psi_{i}$'s to share $\mathcal{O}$ and the optimal solution will in fact be integer.

For this setup the surface $\mathcal{O}$  are geodesics of the manifold. Because the space is negatively curved we are guaranteed that the distance between points on different geodesics is a convex function of the position on each geodesic, parametrized by distance along the geodesic. We consider the family of all possible solutions with two points on $\mathcal{O}$ where the points are labeled by their location on (or affine distance along) $\mathcal{O}$. The objective is the weighted area of the six surfaces connecting these points
\begin{equation}
\area(\Sigma_{2})=\frac{1}{2}(d_{16}+d_{63}+d_{32}+d_{25}+d_{54}+d_{41})
\end{equation}
which is a convex function in all six variables of the locations on $\mathcal{O}$.
We then consider the analogous problem with only one point on each part of $\mathcal{O}$ and furthermore pick the points $1',2',3'$ such that we are at a the global minimum
\begin{equation}
\area(\Sigma)=d_{1'3'}+d_{3'5'}+d_{5'1'}.
\end{equation}
Thus taking $1=2=1',\;3=4=3',\;5=6=5'$ we find this configuration is also a minimum of $\area(\Sigma_{2})$. Since we are at a minimum the gradient is zero thus deforming the surfaces by moving any one of the end points can only increase the value of the objective. Convexity demands that this is the global minimum of $\area(\Sigma_{2})$.

This argument has obvious generalizations to more complex configurations of $\psi_{i}$'s and larger number of boundary regions and shows for these cases \eqref{mprelaxedagain} is exactly dual to \eqref{umemotoprop}. It would be interesting to try to generalize this reasoning to higher dimensions, but we leave this to future work.\footnote{It may possible to prove a form of such a result by utilizing the quotient geometry induced by the relative homology.}

\section{Multipartite surfaces with bulk intersection}\label{sec:bulk}

\begin{figure}[H]
\centering
\includegraphics[width=.5\textwidth]{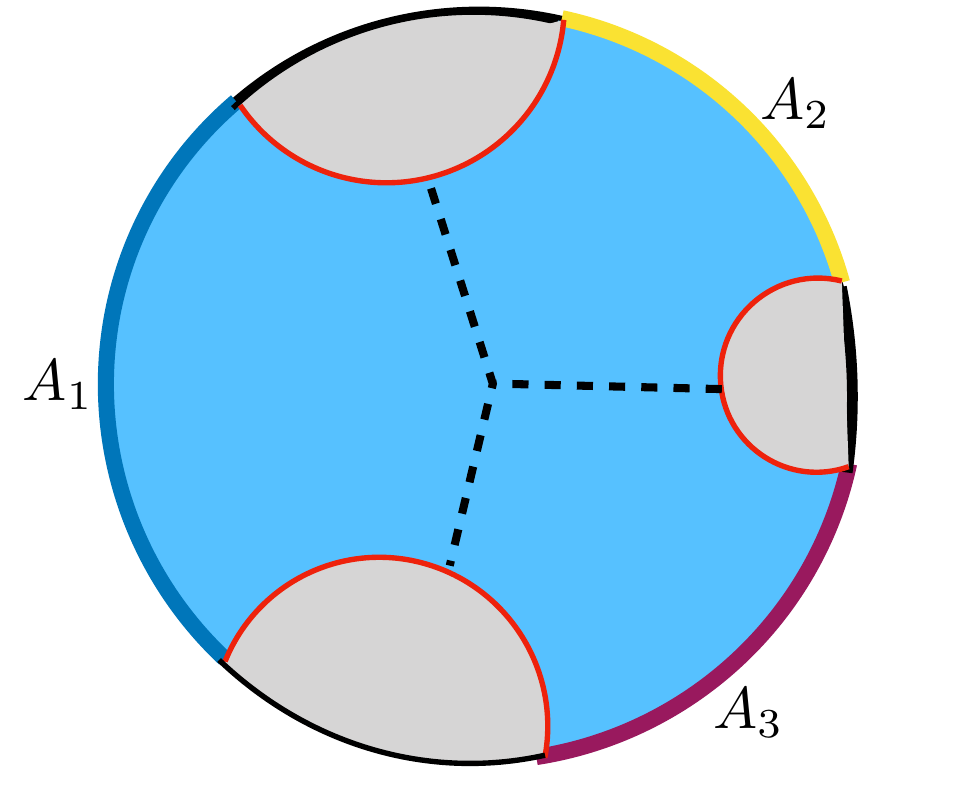}
\caption{\label{fig:mercedes}An example of a class of surfaces in the bulk which allow for bulk intersection points.}
\end{figure}

Another surface similar to the one which gives the holographic multipartite EOP proposal has appeared in the literature \cite{Bao:2018aa}
\begin{equation}
E_{m}(\mathcal{A}) = \min \area(\text{Surfaces which split $r(\mathcal{A})$ into $n$ regions homologous to each $A_{i}$}).
\end{equation}
The key distinction being the allowance of a bulk intersection point (see fig.\ \ref{fig:mercedes}).

Following convex dualization we find some interesting parallels in the flow descriptions between $E_{w}(\mathcal{A)}$ and $E_{m}(\mathcal{A})$.
For simplicity we will dualizing assuming the boundary is pure and then generalize to the case of nonpurity. We dualize by first performing a convex and integer relaxation. We will ignore analogous complications due to the possible inequivalence of the programs due to the integer relaxation. Note that because the minimal surface meets in the bulk we get a constraint on the scalar fields which also must be imposed in the bulk.
\begin{equation}
\min \sum_{i}\int_{M}\sqrt{g} |\nabla_{\mu}\psi_{i}| \quad \text{s.t.} \quad \psi_{i}|_{A_{j}}=\delta_{ij}, \quad \sum_{i}\psi_{i}|_{M} = 1
\end{equation}
We get the Lagrangian
\begin{equation}
\begin{split}
L(\{\psi\},\{\phi\},\{v^{\mu}\},\{\gamma\},\{\beta\}) &= \int_{M}\sqrt{g}\left[ \sum_{i} \left(|\phi_{i}| +v^{\mu i}(\phi_{\mu i}-\nabla_{\mu}\psi_{i}))+ \beta\psi_{i}\right)-\beta\right] \\
&+\sum_{i}\int_{A_{i}}\sqrt{h}\left[\gamma_{ii}(\psi_{i}-1) + \sum_{i \neq j}\gamma_{ij}\psi_{j}\right].
\end{split}
\end{equation}
Dualizing with respect to $\{\phi\}$ we get norm constraints:
\begin{equation} \phi_{\mu i}v^{\mu i} + |\phi_{\mu i}| =0, \quad |v^{\mu i}| \leq 1
\end{equation}
which can be substituted in. After performing integration by parts we have
\begin{equation}
\int_{M}\sqrt{g}\left[\sum_{i}\left(\psi_{i}\nabla_{\mu}v^{\mu i} +\beta\right)-\beta\right] +\sum_{i}\int_{A_{i}}\sqrt{h}\left[\gamma_{ii}(\psi_{i}-1) + \sum_{i \neq j}\gamma_{ij}\psi_{j} + \sum_{j} n_{\mu}v^{\mu j}\psi_{j}\right].
\end{equation}
Next we dualize with respect to $\{\psi\}$.
\begin{equation}
\begin{split}
&\text{On $M$: } \nabla_{\mu}v^{\mu}_{i} =-\beta\\
&\text{On $A_{i}$: } \gamma_{ii} + n_{\mu}v^{\mu}_{i} =0, \quad \gamma_{ij} + n_{\mu}v^{\mu}_{j} =0
\end{split}
\end{equation}
Note that we do not get the normal divergencelessness constraints, flux will be created and destroyed in the bulk. The resulting dual is
\begin{equation}
\max_{\mathcal{V}} \sum_{i}\left(\int_{A_{i}}\sqrt{h} n_{\mu}v^{\mu}_{i} + \frac{1}{n}\int_{M}\sqrt{h}\nabla_{\mu}v^{\mu}_{i}\right).
\end{equation}
We can once again redefine our vector fields
\begin{equation}
w^{\mu} = \frac{1}{n}\sum_{i}v_{i}^{\mu} \quad u^{\mu}_{i} = v^{\mu}_{i}-w^{\mu}
\end{equation}
with the program becoming
\begin{equation}
\max_{\mathcal{U},w} \; \sum_{i}\int_{A_{i}} \sqrt{h}n_{\mu}u^{\mu}_{i} \; \text{s.t. } \nabla_{\mu}u^{\mu}_{i} = 0, \; \nabla_{\mu}w^{\mu} = -\beta, \; |u_{i}+w| \leq 1 , \; \sum_{i}u^{\mu}_{i}=0.
\end{equation}
Comparing to the holographic multipartite EOP proposal essentially what has happened is the entire bulk $M$ has become an unlimited source of $w$ by virtue of the fact that $w$ is not divergenceless and its divergence is an unconstrained function in the bulk when maximizing. It is as if the entire bulk is geometrically connected to the boundary and acts as the purifier. This gives the flows a tremendous amount of freedom when attempting to maximize pushing the minimal surface deep into the bulk. The dualization argument presented can be generalized by also considering a non-pure boundary with region $\mathcal{O}$ as in the case of the multipartite EOP. In doing so all we must change is $w$ in addition to having flux in the bulk can begin and end on $\mathcal{O}$. From this we see starting with this program we can impose $\nabla_{\mu}w^{\mu}=0$ as an explicit constraint and make use of Theorem \ref{T1} to find
\begin{equation}
E_{m}(\mathcal{A}) \geq E_{w}(\mathcal{A}).
\end{equation}

The mechanism for the flows associated with these two surface is almost identical the only difference being whether $w$ is sourced on the boundary, the bulk or both. In general nonzero divergence in the bulk has been associated with quantum corrections to entanglement entropy. While it is not know what information theoretic quantity, if any, this surface is dual to in the boundary CFT, it seems suggestive that quantum corrections may play a role. At the very least these surfaces are intimately connected by the common mechanism for how the corresponding flows are generated. As previously mentioned we expect this to be a general feature of any flow program dual to a bulk surface with intersections.

\section{Squashed entanglement}\label{sec:squashed}
Another notable entanglement measure is the squashed entanglement \cite{Christandl:aa,Tucci:aa} and its multipartite generalization:
\begin{equation}
E_{sq}(\mathcal{A}) = \min I(\mathcal{A} | E).
\end{equation}
Here the minimization is taken over all extensions $\rho_{\mathcal{A}} = \Tr[\rho_{\mathcal{A}E}]$. Most notably the multipartite squashed entanglement has the property
\begin{equation}
E_{sq}(\mathcal{A}) \leq I(\mathcal{A}) \leq E_{p}(\mathcal{A}).
\end{equation}

There are two important points to consider: One, calculating a conditioned entropy measure is implemented using bit threads as a change in the region that a maximal thread configuration can occupy. This can be thought of as imposing additional boundary constraints on the convex program which force the threads to remain in the region. Two, using the surface/state correspondence \cite{Miyaji:2015aa} we can guarantee any choice of extension can be holographically represented as a surface outside the homology region of $\mathcal{A}$. Thus, the most restrictive space we could utilize is $r(\mathcal{A})$ itself. This can also be motivated using MMI which implies utilizing an extension can never be advantageous  \cite{Umemoto2018}.

With these considerations in mind we can make the following observation: during our discussion of the multipartite EOP we were able to explicitly construct a multiflow which calculated $I(\mathcal{A})$ and was contained entirely in the homology region $r(\mathcal{A})$ (see fig.\ \ref{fig:lb}). This shows that holographically the squashed entanglement will be saturated which we view as additional evidence for the conjecture of \cite{Umemoto2018}
\begin{equation}
E_{sq}(\mathcal{A}) = I(\mathcal{A}).
\end{equation}

\bibliographystyle{JHEP}
\bibliography{EoP_BT}
\end{document}